\documentclass[preprint,3p,12pt]{elsarticle}
\def\mathclap{\mathpalette\mathclapinternal}
\def\mathclapinternal#1#2{%
\clap{$\mathsurround=0pt#1{#2}$}}

\usepackage{amssymb}

\usepackage{amsfonts,amsmath,amssymb,amsthm}
\usepackage{bm,nicefrac}
\usepackage{tabularx}
\usepackage{listings}
\usepackage{enumerate}
\usepackage{float}
\usepackage{graphicx}
\usepackage{comment}
\usepackage{subcaption}
\usepackage{paralist}
\usepackage[ruled,lined,algosection]{algorithm2e}
\usepackage[english]{babel}
\usepackage{verbatim}

\usepackage{amsfonts,amsmath,amssymb,amsthm}
\usepackage{mathtools}
\usepackage{graphicx}
\usepackage{hyperref}
\hypersetup{colorlinks,allcolors=blue}

\usepackage{siunitx}
\usepackage{url}
\usepackage{paralist}
\usepackage{adjustbox}

\newtheorem{theorem}{Theorem}

\newtheorem{lemma}{Lemma}

\newcommand{\upeq}[1]{\stackrel{\mathclap{#1}}{=}}

\newcommand{\myref}[2]{\hyperref[#2]{#1~\ref*{#2}}}
\newcommand{\stretched}[2][(s_x,s_y,s_z)]{\overset{\scriptscriptstyle{#1}}{\text{#2}}}
\newcommand{\code}[1]{\texttt{#1}}
\numberwithin{equation}{section}
\usepackage{matlab-prettifier}

\def\clap#1{\hbox to 0pt{\hss#1\hss}}

\def\mathclap{\mathpalette\mathclapinternal}

\def\mathclapinternal#1#2{%
  \clap{$\mathsurround=0pt#1{#2}$}}

\numberwithin{equation}{section}

\begin{document}

\begin{frontmatter}



\title{A GPU-Parallelized Interpolation-Based Fast Multipole Method for the Relativistic Space-Charge Field Calculation}


\author[desy,uhh]{Yi-Kai Kan\corref{author}}
\author[desy,uhh]{Franz X. Kärtner}
\author[tuhh]{Sabine Le Borne}
\author[tuhh]{Jens-Peter M. Zemke}

\address[desy]{Center for Free-Electron Laser Science CFEL, Deutsches Elektronen-Synchrotron DESY, Germany}
\address[uhh]{Department of Physics, University of Hamburg, Germany}
\address[tuhh]{Hamburg University of Technology, Institute of Mathematics, Germany}
\cortext[author] {Corresponding author.\\\textit{E-mail address:} yikai.kan@desy.de (Y.-K. Kan)}

\begin{abstract}
The fast multipole method (FMM) has received growing attention in the beam physics simulation. In this study, we formulate an interpolation-based FMM for the computation of the relativistic space-charge field. Different to the quasi-electrostatic model, our FMM is formulated in the lab-frame and can be applied without the assistance of the Lorentz transformation. In particular, we derive a modified admissibility condition which can effectively control the interpolation error of the proposed FMM. The algorithms and their GPU parallelization are discussed in detail. A package containing serial and GPU-parallelized solvers is implemented in the Julia programming language. The GPU-parallelized solver can reach a speedup of more than a hundred compared to the execution on a single CPU core.
\end{abstract}

\begin{keyword}
fast multipole method\sep space-charge field calculation\sep separable approximation\sep admissibility condition\sep GPU parallelization.
\end{keyword}

\end{frontmatter}


\section{Introduction}
The space-charge effect is one of the most important topics in the study of beam physics. In recent years, the fast multipole method (FMM) has attracted increasing attention in the numerical modeling of the space-charge field~\cite{zhang2015differential,zhang2017fast,langston2021mach,gordon2021point,schmid2019simulating}. Compared to the particle-in-cell (PIC) method~\cite{dawson1983particle,birdsall2018plasma,flottmann2003recent,qiang2006three-dimensional,qiang2018symplectic} which has been a standard choice in the community of accelerator physics for decades, the FMM has the following advantages in the context of the study of beam physics:
\begin{compactenum}
    \item Through the point-to-point (P2P) operation, FMMs inherently consider the point-to-point Coulomb effects, \emph{e.g.}, disorder-induced heating and the Boersch effect which are non-negligible in the simulation for cold and dense beams~\cite{gordon2021point}.
    \item The FMM is a gridless algorithm and can effectively handle charged particle beams with complex geometry~\cite{zhang2017fast}. 
\end{compactenum}
There have been many efforts using FMM for the modeling of electrostatic 
Coulomb interactions~\cite{jones1998hybrid,zhang2015differential,gordon2021point,langston2021mach}. The electrostatic model is suitable for the study of non-relativistic particle beams, \emph{e.g.}, the simulation of ultrafast electron microscopy~\cite{zhang2015differential,gordon2021point} and the simulation of proton dynamics in synchrotrons~\cite{jones1998hybrid}. However, for the modeling of energetic electron beams, the consideration of the relativistic effect on the particle field may be essential. One approach to include relativistic effects is the quasi-electrostatic model~\cite{schmid2018reptil,schmid2019simulating} where all particles are assumed stationary in a rest-frame of the particle beam. The electrostatic field on each particle is first solved in the rest-frame and the corresponding field in the lab-frame is calculated through the Lorentz transformation~\cite{flottmann2003recent,qiang2017symplectic}. Because of the assumption made in the quasi-electrostatic model, some adjustments are necessary to incorporate the effect of the momentum spread. For example, to handle particle beams with larger energy spread, a technique called energy-binning was proposed by binning particles in energy; the total space-charge field comes from the superposition of the source particle field evaluated in the rest-frame of each bin~\cite{schmid2021energy,fubiani2006space}.   This study is the extension of our previous work on treecode~\cite{kan2023relativistic} and consists of two parts. In the first part (Sections~\ref{sec:fmm_idea_of_fmm}--\ref{sec:fmm_result}), we formulate an FMM for the efficient computation of the relativistic space-charge field. Our formulation is based on the barycentric Lagrange dual tree traversal (BLDTT) proposed in Ref.~\cite{wilson2021gpu-accelerated}. BLDTT uses barycentric Lagrange interpolation for the kernel approximation and dual tree traversal for the construction of interaction lists. Different from the quasi-electrostatic model, the proposed FMM is formulated in the lab-frame and can be applied without the use of the Lorentz transformation. We first introduce the idea of an interpolation-based FMM. After that, we formulate an interpolation-based FMM for the computation of the relativistic space-charge field. In particular, we derive a modified admissibility condition for the cluster-cluster interaction of the relativistic kernel used in the formulated FMM. The algorithms and the implementation details associated with the proposed FMM are also provided. The second part of this study (Sections~\ref{sec:fmm_array_tree}--\ref{sec:fmm_performance}) is devoted to the GPU parallelization of the proposed FMM. Different to Ref.~\cite{wilson2021gpu-accelerated} which is based on OpenACC (a directive-based programming model)~\cite{Wienke2012openacc}, our parallelization is based on the CUDA programming model.   We discuss the data structure and the design issues in the implementation of the GPU parallelization. A package containing serial and GPU-parallelized solvers is implemented in the Julia programming language. The performance of the parallel solver is also demonstrated.

\section{The Idea of FMM}\label{sec:fmm_idea_of_fmm}
In this section, we give a short overview of the interpolation-based FMM~\cite{fong2009black,wilson2021gpu-accelerated}. Consider two particle-clusters $S_t$ and $S_s$. The total force-field $f$ from the source particles in the cluster $S_s$ applied on a target particle with index $i$ and position $\bm{x}_i\in S_t$ through an interaction kernel $g(\bm{\cdot,\cdot})$ can be modeled as
\begin{equation}\label{eq:nbody_summation}
    f(\bm{x}_i) = \sum_{j\in \widehat{S}_s} g(\bm{x}_i, \bm{x}_j)m_j.
\end{equation}
Here and in the following, $\widehat{S}_t$ and $\widehat{S}_s$ represent the index sets of the particles in $S_t$ and $S_s$, respectively. The symbol $m_j$ is the physical quantity of the source particle with index $j$. Although the actual meaning of $m_j$ depends on the physics problem we investigate, without loss of generality, we call it mass throughout this section.

The idea of FMM for a fast evaluation of the summation in~\eqref{eq:nbody_summation} is based on an approximation of the kernel function by interpolating both the target variable $\bm{x}_i$ and the source variable $\bm{x}_j$
\begin{equation}\label{eq:interpolation_approx_both_variable}
    \bm{g}(\bm{x}_i,\bm{x}_j)\approx\sum_{\bm{\mu}}\sum_{\bm{\nu}}\ell_{S_t,\bm{\mu}}(\bm{x}_i)\bm{g}(\bm{\xi}_{S_t,\bm{\nu}}, \bm{\xi}_{S_s,\bm{\nu}})\ell_{S_s,\bm{\nu}}(\bm{x}_j).
\end{equation}
Here, we define the bounding box $Q$ of a cluster $S$ as
\begin{align*}
    &Q=[a_{x},b_{x}]\times[a_{y},b_{y}]\times[a_{z},b_{z}] \subset \mathbb{R}^3
\end{align*}
with $a_{g}=\min_{i\in\widehat{S}}\{g\}$, $b_{g}=\max_{i\in\widehat{S}}\{g\}$ and $g\in\{x,y,z\}$. 
The Lagrange basis polynomials over the bounding boxes of $S_t$ and $S_s$ are defined as
\begin{align*}
    &\ell_{S_t,\bm{\mu}}(\bm{x}_i):=\ell^{x}_{S_t,\mu_1}(x_i)\cdot\ell^{y}_{S_t,\mu_2}(y_i)\cdot\ell^{z}_{S_t,\mu_3}(z_i), \\
    &\ell_{S_s,\bm{\nu}}(\bm{x}_j):=\ell^{x}_{S_s,\nu_1}(x_j)\cdot\ell^{y}_{S_s,\nu_2}(y_j)\cdot\ell^{z}_{S_s,\nu_3}(z_j),
\end{align*}
with the corresponding interpolation points
\begin{equation*}
    \bm{\xi}_{S_t\,\bm{\mu}}:=(\xi_{S_t,\mu_1},\xi_{S_t,\mu_2},\xi_{S_t,\mu_3}) 
    \quad\text{and}\quad 
    \bm{\xi}_{S_s\,\bm{\nu}}:=(\xi_{S_s,\nu_1},\xi_{S_s,\nu_2},\xi_{S_t,\nu_3}).
\end{equation*}
Substituting \eqref{eq:interpolation_approx_both_variable} into \eqref{eq:nbody_summation}, we have
\begin{equation}\label{eq:FMM_scheme}
\begin{split}
    f(\bm{x}_i)
    &=\sum_{j\in \widehat{S}_s} g(\bm{x}_i, \bm{x}_j)m_j \\
    &\approx \sum_{\bm{\mu}}\ell_{S_t,\bm{\mu}}(\bm{x}_i)\sum_{\bm{\nu}}\bm{g}(\bm{\xi}_{S_t,\bm{\mu}}, \bm{\xi}_{S_s,\bm{\nu}})
    \underbrace{\sum_{j\in\widehat{S}_s}\ell_{S_s,\bm{\nu}}(\bm{x}_j)m_j}_{=:M_{S_s,\bm{\nu}}}\\
    &=\sum_{\bm{\mu}}\ell_{S_t,\bm{\mu}}(\bm{x}_i)\underbrace{\sum_{\bm{\nu}}\bm{g}(\bm{\xi}_{S_t,\bm{\mu}},\bm{\xi}_{S_s,\bm{\nu}})M_{S_s,\bm{\nu}}}_{=:L_{S_t,\bm{\mu}}}\\
    &=\sum_{\bm{\mu}}\ell_{S_t,\bm{\mu}}(\bm{x}_i)L_{S_t,\bm{\mu}}.
\end{split}
\end{equation}
By observing \eqref{eq:FMM_scheme}, we identify the four of FMM kernels:
\begin{compactitem}
\item P2M (point to multipole): the micro particles in the cluster $S_s$ are aggregated into a few macro particles and the mass of each macro particle ($M_{S_s,\bm{\nu}}$, also called multipole) can be computed by
\begin{equation}\label{eq:P2M}
    M_{S_s,\bm{\nu}}:=\sum_{j\in\widehat{S}_s}\ell_{S_s,\bm{\nu}}(\bm{x}_j)m_j.
\end{equation}
\item M2L (multipole to local): the multipoles of the source cluster are used to evaluate the force-fields acting on the macro particles ($L_{S_t,\bm{\mu}}$, also called local field) in the target cluster $S_t$
\begin{equation}\label{eq:M2L}
    L_{S_t,\mu}:= \sum_{\bm{\nu}}\bm{g}(\bm{\xi}_{S_t,\bm{\mu}}, \bm{\xi}_{S_s,\bm{\nu}})M_{S_s,\bm{\nu}}.
\end{equation}
\item L2P (local to point): in the target cluster, the effective force-fields acting on the macro particles are transferred to the micro particle at $\bm{x}_i$ by
\begin{equation}\label{eq:L2P}
    f(\bm{x_i})=\sum_{\bm{\mu}}\ell_{S_t,\bm{\mu}}(\bm{x}_i)L_{S_t,\bm{\mu}}.
\end{equation}
\item P2P (point to point): if $S_t$ and $S_s$ do not fulfill an admissibility condition (ADMC, also called multipole acceptance criteria MAC in some literature~\cite{wang2020kernel,wilson2021development}) so that~\eqref{eq:interpolation_approx_both_variable} is not applicable, the force-field can be calculated directly by
\begin{equation}\label{eq:P2P}
    f(\bm{x}_i)=\sum_{j\in \widehat{S}_s} g(\bm{x}_i, \bm{x}_j)m_j\quad\forall i\in \widehat{S}_t.
\end{equation}
This also applies for the case $S_t=S_s=S$, where $i,\,j\in \widehat{S}$ and $i\neq j$.
\end{compactitem}
One main feature of the FMM is the consideration of the cluster-cluster interaction (M2L) through macro particles; and hence, the total number of operations for the evaluation of force-fields can be drastically reduced. In the FMM, we first partition all particles in the system into a hierarchy of clusters (cluster tree). If we directly use~\eqref{eq:P2M} to compute the multipole of each cluster, the number of operations for computing the multipoles of the whole cluster tree is
\begin{equation*}
    \text{const}\cdot N\cdot N_r\cdot\log\left(\dfrac{N}{N_r}\right)
\end{equation*}
with the assumption that the number of macro particles used for the approximation and the number of micro particles contained in the leaf cluster are both $N_r$. To reduce the total number of operations for computing the multipoles, we can make use of the following property of polynomial interpolation stated in \myref{Theorem}{th:good_polynomial_property}. 
\begin{theorem}\label{th:good_polynomial_property}
If $P(x)$ is a polynomial function of degree $n$, we have
\begin{equation}\label{eq:interpolation_of_polynormial}
    P(x)=\sum^{n}_{k=0}P(\xi_{S,k})\ell_{S,k}(x)\quad \forall x\in S
\end{equation}
with $\ell_{S,k}(x)$ denoting the Lagrange basis for polynomials of degree~$\leq n$ for the interpolation point $\xi_{S,k}\in S$.
\end{theorem}
This equality can be seen by the fundamental theorem of algebra since both LHS and RHS have the same values at the $n+1$ points $\{\xi_{S,k} \mid k=0,\ldots,n\}$ and the RHS is a polynomial of degree $n$.
By using \myref{Theorem}{th:good_polynomial_property}, we can introduce two further procedures and two kernels of FMM:
\begin{compactitem}
    \item Upward Pass: a source particle cluster $S$ is subdivided into a hierarchy of clusters of the depth~$\kappa$ (called cluster tree~\cite{boerm2010efficient,hackbusch2015hierarchical}) . The multipoles of each cluster can be computed by the multipoles of its children clusters in view of
    \begin{align}
        M_{S,\bm{\nu}}
        &=\sum_{j\in\widehat{S}}\ell_{S,\bm{\nu}}(\bm{x}_j)m_j \nonumber\\
        &=\sum_{s'\in\text{children}(S)}\sum_{j\in\widehat{s'}}\ell_{S,\bm{\nu}}(\bm{x}_j)m_j \nonumber\\
        &\upeq{\eqref{eq:interpolation_of_polynormial}}\sum_{s'\in\text{children}(S)}\sum_{j\in\widehat{s'}}\sum_{\bm{\nu}'}\ell_{S,\bm{\nu}}(\xi_{s',\bm{\nu}'})\ell_{s',\bm{\nu}'}(\bm{x}_j)m_j \nonumber\\
        &=\sum_{s'\in\text{children}(S)}\sum_{\bm{\nu}'}\ell_{S,\bm{\nu}}(\xi_{s',\bm{\nu}'})\sum_{j\in\widehat{s'}}\ell_{s',\bm{\nu}'}(\bm{x}_j)m_j \nonumber\\
        &=\sum_{s'\in\text{children}(S)}\sum_{\bm{\nu}'}\ell_{S,\bm{\nu}}(\xi_{s',\bm{\nu}'})M_{s',\bm{\nu}'}. \label{eq:M2M}
    \end{align}
    Equation~\eqref{eq:M2M} is the formula of the M2M (multipole to multipole) kernel. In FMM, the multipoles of the source cluster tree are updated by a procedure called upward pass. In the upward pass, the multipoles of the leaf clusters are first evaluated with P2M~\eqref{eq:P2M}; and then, we start from the second deepest level of the cluster tree (\emph{i.e.}, level $\kappa-1$) and apply M2M~\eqref{eq:M2M} to compute the multipoles of each cluster in each level (level by level). 
    \item Downward Pass: a target cluster $S$ is subdivided into a cluster tree of the depth $\kappa$ and each target particle (say particle~$i$) will be contained in a sequence of clusters $\{S^{l} \mid l=0,\dots,\kappa \}$ from each level $l$ with $S^{l+1}\subset S^{l}$ and $S^{0}=S$. The force-field on the target particle~$i$ can be calculated by
    \begin{equation}\label{eq:downwardpass_L2P}
        f(\bm{x}_i)=\sum^{\kappa}_{l=0}\sum_{\bm{\mu}}\ell_{S^{l},\bm{\mu}}(\bm{x}_i)L_{S^{l},\bm{\mu}} =\sum_{\bm{\mu}}\mathcal{L}_{S^{\kappa},\bm{\mu}}\ell_{S^{\kappa},\bm{\mu}}(\bm{x}_i).
    \end{equation}
    Here, we define the “cumulative local field” $\mathcal{L}_{S^{l}}$, which follows the recursive relation 
    \begin{equation}\label{eq:L2L}
        \mathcal{L}_{S^{l},\bm{\mu}}:= L_{S^{l},\bm{\mu}}+ 
        \sum_{\bm{\mu}'}\mathcal{L}_{S^{l-1},\mu'}\cdot\ell_{S^{l-1},\bm{\mu}'}(\bm{\xi}_{S^{l},\bm{\mu}})
        \quad\text{with}\quad
        \mathcal{L}_{S^{0},\bm{\mu}}:=L_{S^{0},\bm{\mu}}.
    \end{equation}
    Equation~\eqref{eq:downwardpass_L2P} can be proved by using~\eqref{eq:interpolation_of_polynormial}, \eqref{eq:L2L} and mathematical induction (cf.~\myref{Section}{appendix:fmm_proof_cumulative_local}). Therefore, during the downward pass of FMM, we first perform L2L~\eqref{eq:L2L} to calculate the cumulative local fields of the deepest-level cluster $S^{\kappa}$; afterward, we transfer $\mathcal{L}_{S^{\kappa},\mu}$ to the target particles contained in $S^{\kappa}$ via L2P~\eqref{eq:downwardpass_L2P}.
\end{compactitem}

\section{FMM Formulation for the Efficient Computation of Relativistic Space-Charge Field}\label{sec:fmm_formulation}
Consider a relativistic charged particle beam moving in $z$-direction. Inside the particle beam, the space-charge field from a source particle with the position $\bm{x}_j$ exerting to a target particle with the position $\bm{x}_i$ can be approximately written as~\cite{kan2023relativistic}
\begin{equation}\label{eq:fmm_pp_force_field}
    \bm{E}(\bm{x}_i, \bm{x}_j) \approx 
    \dfrac{q}{4\pi\epsilon_0}
    \gamma_j\bm{g}(\bm{x}, \bm{x}_j)
    \quad\text{and}\quad 
    \bm{B}(\bm{x}_i, \bm{x}_j) \approx  
    \dfrac{q}{4\pi\epsilon_0 c_0}
    \bm{p}_j\times \bm{g}(\bm{x}, \bm{x}_j),
\end{equation}
with the kernel function called “relativistic kernel”
\begin{equation}\label{eq:fmm_relativistic_kernel}
    \bm{g}(\bm{x}_i,\bm{x}_j):=\dfrac{\bm{x}_i-\bm{x}_j}{\biggl((x_i-x_j)^2+(y_i-y_j)^2+\overline{\gamma}^{2}(z_i-z_j)^2\biggr)^{3/2}},
\end{equation}
where $\gamma_j=1/\sqrt{1-\Vert\bm{\beta}_j\Vert^2_2}$ and $\bm{p}_j=\gamma_j\bm{\beta}_j$ are the Lorentz factor and normalized momentum of the particle, respectively, with $\bm{\beta}_j$ the particle velocity normalized to the speed of light $c_0$. Here, $\overline{\gamma}$ is the average Lorentz factor and can be computed by $\overline{\gamma}^2=1+\overline{\bm{p}}\cdot\overline{\bm{p}}$ with the average momentum of the particle beam $\overline{\bm{p}}$. Throughout this study, we assume that all particles in the particle beam are the same type with charge $q$.

Given a target particle with the position $\bm{x}_i$ contained in a target cluster $S_t$, the space-charge field from all the particles in the cluster $S_s$ experienced by this target particle can be computed approximately by applying \eqref{eq:interpolation_approx_both_variable} to \eqref{eq:fmm_pp_force_field},
\begin{align}\label{eqs:fmm_cluster_field}
\begin{split}
    \sum_{j\in\widehat{S}_s}\bm{E}(\bm{x}_i, \bm{x}_j)
    &\approx\dfrac{q}{4\pi\epsilon_0}\sum_{\bm{\mu}}\sum_{\bm{\nu}}\sum_{j\in \widehat{S}_s}\ell_{S_t,\bm{\mu}}(\bm{x}_i)\ell_{S_s,\bm{\nu}}(\bm{x}_j)\gamma_{j}\bm{g}(\bm{\xi}_{S_t,\bm{\mu}}, \bm{\xi}_{S_s,\bm{\nu}}) \\
    &=\dfrac{q}{4\pi\epsilon_0}\sum_{\bm{\mu}}\ell_{S_t,\bm{\mu}}(\bm{x}_i)\bm{E}_{S_t,\bm{\mu}}, \\
    \sum_{j\in\widehat{S}_s}\bm{B}(\bm{x}_i, \bm{x}_j)
    &\approx\dfrac{q}{4\pi\epsilon_0 c_0}
    \sum_{\bm{\mu}}\sum_{\bm{\nu}}\sum_{j\in \widehat{S}_s}\ell_{S_t,\bm{\mu}}(\bm{x}_i)\ell_{S_s,\bm{\nu}}(\bm{x}_j)\bm{p}_{j}\times\bm{g}(\bm{\xi}_{S_t,\bm{\mu}}, \bm{\xi}_{S_s,\bm{\nu}})\\
    &=\dfrac{q}{4\pi\epsilon_0 c_0}\sum_{\bm{\mu}}\ell_{S_t,\bm{\mu}}(\bm{x}_i)\bm{B}_{S_t,\bm{\mu}}.
\end{split}
\end{align}
Here, we introduce the effective Lorentz factor and the effective momentum of a macro particle with the position $\bm{\xi}_{S_s,\bm{\nu}}$ and index $\bm{\nu}$ in the cluster $S_s$ as
\begin{equation*}\label{eq:fmm_effective_mass}
    \gamma_{S_s,\bm{\nu}}:=\sum_{j\in \widehat{S}_s}\ell_{S_s,\bm{\nu}}(\bm{x}_j)\gamma_{j}
    \quad\text{and}\quad
    \bm{p}_{S_s,\bm{\nu}}:=\sum_{j\in \widehat{S}_s}\ell_{S_s,\bm{\nu}}(\bm{x}_j)\bm{p}_{j}.
\end{equation*}
 Similarly, the effective electric and magnetic fields experienced by a macro particle with the position $\bm{\xi}_{S_t,\bm{\mu}}$ and index $\bm{\mu}$ in the target cluster $S_t$ are defined as
\begin{equation*}
    \bm{E}_{S_t,\bm{\mu}}:=\sum_{\bm{\nu}}\gamma_{S_s,\bm{\nu}}\bm{g}(\bm{\xi}_{S_t,\bm{\mu}}, \bm{\xi}_{S_s,\bm{\nu}})
    \quad\text{and}\quad
    \bm{B}_{S_t,\bm{\mu}}:=\sum_{\bm{\nu}}\bm{p}_{S_s,\bm{\nu}}\times\bm{g}(\bm{\xi}_{S_t,\bm{\mu}}, \bm{\xi}_{S_s,\bm{\nu}}).
\end{equation*}

\section{Admissibility Condition for Cluster-Cluster Interaction of the Relativistic Kernel}\label{sec:fmm_ADMC}
In the previous section, we used Lagrangian interpolation to approximate the space-charge field on a target particle in a relativistic particle beam. It is also of importance to know when this approximation can be applied. To answer this question, we may investigate the interpolation error bound of the relativistic kernel  
\begin{equation}\label{eq:fmm_interpolation_error}
    \Vert \bm{g}(\bm{x}_i,\bm{x}_j) - \widetilde{\bm{g}}(\bm{x}_i,\bm{x}_j) \Vert_{\infty,Q_{t}\times Q_{s}}
    \leq B_{t} + B_{s}.
\end{equation}
Here, $B_{t}$ and $B_{s}$ are the interpolation error bounds with respect to the target variable $\bm{x}_i:=(x_i, y_i, z_i)$ and the source variable $\bm{x}_j:=(x_j, y_j, z_j)$:
\begin{align*}
    B_{t}&:=\text{const}\cdot\quad\sum_{\mathclap{k\in\{x_i,y_i,z_i\}}}\quad (b_k - a_k)^{n+1}\cdot\dfrac{\Vert \partial^{n+1}_{k}\bm{g}(\bm{x}_i,\bm{x}_j)\Vert_{\infty,Q_{t}\times Q_{s}}}{(n+1)!}, \\
    B_{s}&:=\text{const}\cdot\quad\sum_{\mathclap{k\in\{x_j,y_j,z_j\}}}\quad (b_k - a_k)^{n+1}\cdot\dfrac{\Vert \partial^{n+1}_{k}\bm{g}(\bm{x}_i,\bm{x}_j)\Vert_{\infty,Q_{t}\times Q_{s}}}{(n+1)!},
\end{align*}
where the bounding boxes of $S_t$ and $S_s$ are
\begin{align*}
    &Q_t=[a_{x_i},b_{x_i}]\times[a_{y_i},b_{y_i}]\times[a_{z_i},b_{z_i}] \subset \mathbb{R}^3, \\
    &Q_s=[a_{x_j},b_{x_j}]\times[a_{y_j},b_{y_j}]\times[a_{z_j},b_{z_j}] \subset \mathbb{R}^3.
\end{align*}

Following the similar analysis in Ref.~\cite{kan2023relativistic}, we can derive the error bounds of $B_{t}$ and $B_{s}$ respectively as
\begin{align*}
    B_{t}&\leq
    \text{const}\cdot\quad\sum_{\mathclap{k\in\{x_i,y_i,z_i\}}}\quad 
    s^{n+1}_k(b_k - a_k)^{n+1}\cdot\dfrac{1}{\Vert           \bm{s}\circ(\bm{x}_i-\bm{x}_j)\Vert^{n+3}_{\infty}} 
    \leq
    \dfrac{\text{const}}{\stretched[(1,1,\overline{\gamma})]{dist}(S_t,S_s)^2}\dfrac{\stretched[(1,1,\overline{\gamma})]{diam}(S_t)^{n+1}}{\stretched[(1,1,\overline{\gamma})]{dist}(S_t,S_s)^{n+1}}, \\
    B_{s} &\leq
    \text{const}\cdot\quad\sum_{\mathclap{k\in\{x_j,y_j,z_j\}}}\quad 
    s^{n+1}_k(b_k - a_k)^{n+1}\cdot\dfrac{1}{\Vert           \bm{s}\circ(\bm{x}_i-\bm{x}_j)\Vert^{n+3}_{\infty}}
    \leq
    \dfrac{\text{const}}{\stretched[(1,1,\overline{\gamma})]{dist}(S_t,S_s)^2}\dfrac{\stretched[(1,1,\overline{\gamma})]{diam}(S_s)^{n+1}}{\stretched[(1,1,\overline{\gamma})]{dist}(S_t,S_s)^{n+1}},
\end{align*}
where we define a stretched vector $\bm{s}:=(s_x,s_y,s_z)=(1,1,\overline{\gamma})$ with the average Lorentz factor $\overline{\gamma}$ in \eqref{eq:fmm_relativistic_kernel}. The symbol $\circ:\mathbb{R}^n\times\mathbb{R}^n\to\mathbb{R}^n$ denotes the component-wise product of two vectors. Here, the stretched diameter of a cluster and the stretched distance between two clusters $S_t$ and $S_s$ are
\begin{equation}\label{eq:fmm_stretched_norm_dist_def}
\begin{split}
    \stretched{diam}(S)
    &:=\max_{\bm{x}_i,\bm{x}_j\in S} \Vert(s_x,s_y,s_z)\circ(\bm{x}_i-\bm{x}_j)\Vert_2, \\
    \overset{\scriptscriptstyle{(s_x,s_y,s_z)}}{\text{dist}}(S_t,S_s)
    &:=\min_{{\substack{\bm{x}_i \in S_t \\ \bm{x}_j \in S_s}}} \Vert(s_x,s_y,s_z)\circ(\bm{x}_i-\bm{x}_j)\Vert_2.
\end{split}
\end{equation}
Therefore, the interpolation error~\eqref{eq:fmm_interpolation_error} can be bounded by
\begin{equation*}
\begin{split}
    \Vert \bm{g}(\bm{x}_i,\bm{x}_j) - \widetilde{\bm{g}}(\bm{x}_i,\bm{x}_j) \Vert_{\infty,Q_{t}\times Q_{s}}
    &\leq
    \text{const}\cdot\dfrac{\stretched[(1,1,\overline{\gamma})]{diam}(S_t)^{n+1}+\stretched[(1,1,\overline{\gamma})]{diam}(S_s)^{n+1}}{\stretched[(1,1,\overline{\gamma})]{dist}(S_t,S_s)^{n+1}} \\
    &\leq
    \text{const}\cdot\left(\dfrac{\stretched[(1,1,\overline{\gamma})]{diam}(S_t)+\stretched[(1,1,\overline{\gamma})]{diam}(S_s)}{\stretched[(1,1,\overline{\gamma})]{dist}(S_t,S_s)}\right)^{n+1} \\
    &\leq
    \text{const}\cdot\left(\dfrac{
    \max\biggl(\stretched[(1,1,\overline{\gamma})]{diam}(S_t),\stretched[(1,1,\overline{\gamma})]{diam}(S_s)\biggr)}{\stretched[(1,1,\overline{\gamma})]{dist}(S_t,S_s)}\right)^{n+1},
\end{split}
\end{equation*}
and we can define an admissibility condition for the cluster-cluster interaction of the relativistic kernel by
\begin{equation}\label{eq:fmm_stretched_admc}
    \dfrac{\max\biggl(\stretched[(1,1,\overline{\gamma})]{diam}(S_t),\stretched[(1,1,\overline{\gamma})]{diam}(S_s)\biggr)}{\stretched[(1,1,\overline{\gamma})]{dist}(S_t,S_s)} < \eta
\end{equation}
with some admissibility parameter $\eta\in\mathbb{R}_{>0}$ which can be chosen to control the interpolation error bound.

Besides deriving the stretched admissibility condition for the relativistic kernel, it is possible to bypass the mathematical analysis by using the Lorentz transformation. In the rest-frame of a particle beam with $\overline{\bm{p}}=0$, the relativistic kernel $\bm{g}$ is approximately equal to the electrostatic kernel and the conventional admissibility condition can be used for controlling the interpolation error. However, as already discussed in our previous work~\cite{kan2023relativistic}, this approach can result in a larger error when a particle beam with larger momentum spread is considered because the distance of each target-source pair in the rest-frame of the particle beam is not correctly evaluated. Therefore, we will not discuss this approach in this study.      

\section{The Procedure of FMM}\label{sec:fmm_procedure}
With the FMM kernels introduced in \myref{Section}{sec:fmm_idea_of_fmm} and the stretched admissibility condition for the cluster-cluster interaction of the relativistic kernel derived in \myref{Section}{sec:fmm_ADMC}, we can formulate an FMM for the calculation of relativistic space-charge field.
The proposed FMM~(\myref{Algorithm}{alg:FMM}) can be summarized in the following procedures: 
\begin{compactenum}
    \item A cluster tree using a k-d tree with a cardinality-balanced subdivision scheme~\cite[Section~3]{kan2023relativistic} is constructed from the particles in the system. To control the interpolation error subject to the relativistic kernel~\eqref{eq:fmm_relativistic_kernel} with the stretched admissibility condition~\eqref{eq:fmm_stretched_admc}, the effect of the stretch also needs to be considered in the particle cluster subdivision~\cite{kan2023relativistic}. Hence, the cluster tree is constructed through \myref{Algorithm}{alg:fmm_subdivide_cluster} with $\bm{s}=(1,1,\overline{\gamma})$.
    \item In the upward pass~(\myref{Algorithm}{alg:upwardpass}), the multipoles of each leaf cluster is computed by P2M~(\myref{Algorithm}{alg:P2M}) and are transferred to the multipoles of its ascendants by M2M~(\myref{Algorithm}{alg:M2M}).
    \item A list of interaction pairs is determined dynamically by the dual tree traversal~\cite{appel1985efficient,dehnen2002hierarchical}. The corresponding pseudocode is \myref{Algorithm}{alg:dualtraverse}. In the simulation of the space-charge effect for a particle beam, the roots of the target tree $S_t$ and the source tree $S_s$ are an identical cluster $S$, \emph{i.e.}, $S_t=S_s=S$. For the pair of clusters fulfilling the admissibility condition~(\myref{Algorithm}{alg:stretched_admissible_CC} with $\bm{s}=(1,1,\overline{\gamma})$), the local field of the target cluster is computed by M2L~(\myref{Algorithm}{alg:M2L}). For the interaction pair of leaf clusters, the force-fields on the particles in the target cluster are computed by P2P~(\myref{Algorithm}{alg:P2P}). 
    In our implementation, instead of using~\eqref{eq:fmm_stretched_norm_dist_def}, we adapt a different definition to compute the stretched diameter and the stretched distance as illustrated respectively in \myref{Algorithm}{alg:stretched_diam} and \myref{Algorithm}{alg:stretched_dist} because of their simplicity in the practical implementation.
    \item In the downward pass~(\myref{Algorithm}{alg:downwardpass}), the cumulative local fields of each cluster are transferred to its descendants by L2L~(\myref{Algorithm}{alg:L2L}) and the cumulative local fields of each leaf cluster are transferred to the force-field on its member particles by L2P~(\myref{Algorithm}{alg:L2P}).
\end{compactenum}
A schematic comparison of treecode~\cite{wang2020kernel} and FMM is illustrated in \myref{Figure}{fig:fmm_schematic_treecode_vs_fmm}. In the treecode (\myref{Figure}{fig:fmm_schematic_treecode_vs_fmm_a}), we interpolate the source variable $\bm{x}_j$ of the kernel function so that we can cluster the source particles and build up the effective masses of each cluster. The force-field of each particle in the target cluster is evaluated by each independent traversal of the source cluster tree and by investigating the particle-cluster interaction. In the FMM~(\myref{Figure}{fig:fmm_schematic_treecode_vs_fmm_b}), we interpolate both target variable $\bm{x}_i$ and source variable $\bm{x}_j$ of the kernel function so that the target particles and source particles can be clustered and the corresponding effective force-fields and masses of each cluster can be built up. The force-field on a target particle is transferred from the effective force-fields of the clusters, which are computed by cluster-cluster interaction through the traversal of the target cluster tree and the source cluster tree simultaneously.   
\begin{figure}[h]
\centering
\begin{subfigure}[b]{0.495\textwidth}
    \includegraphics[width=\linewidth]{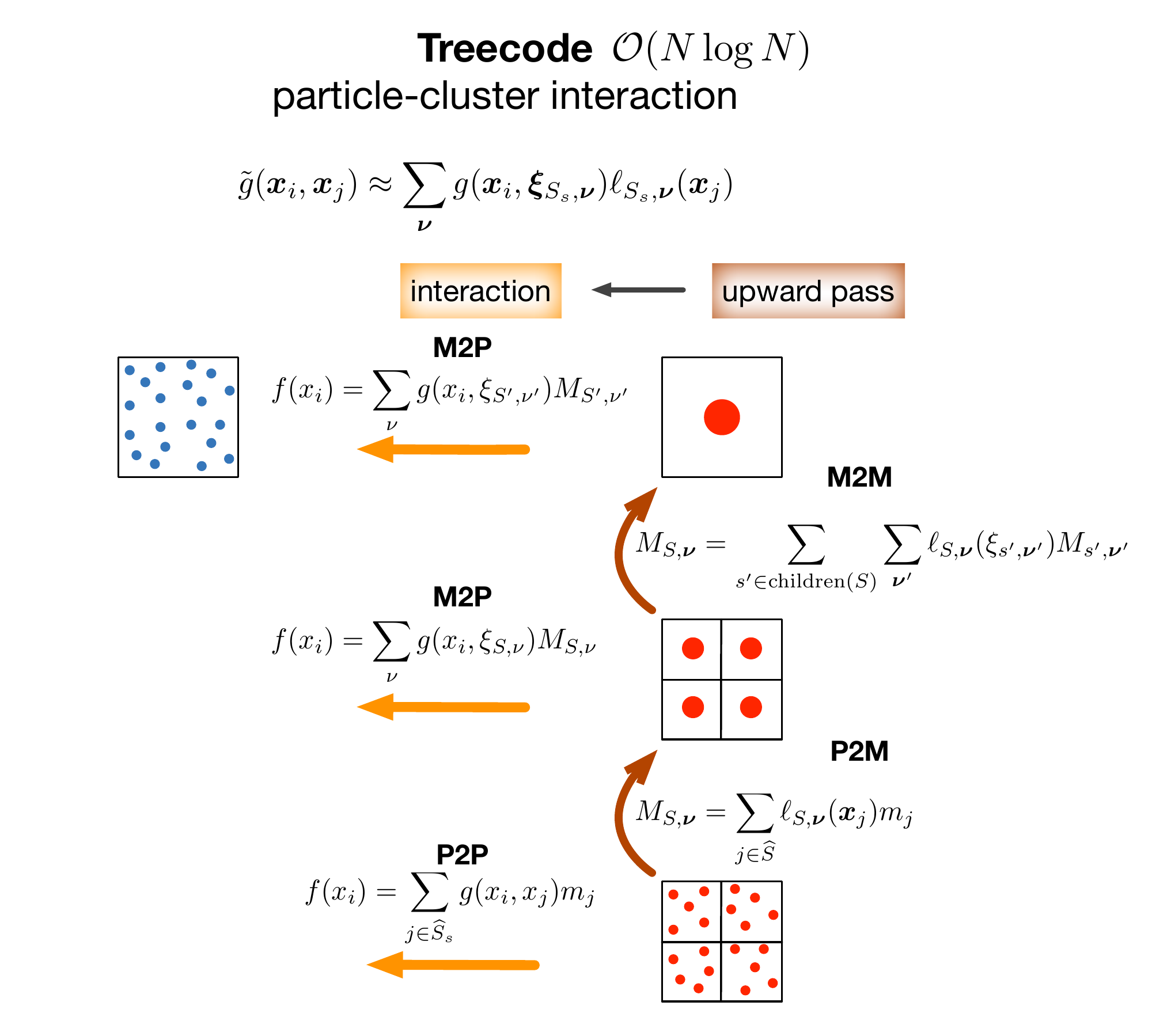}
    \caption{\label{fig:fmm_schematic_treecode_vs_fmm_a}Treecode}
\end{subfigure}
\begin{subfigure}[b]{0.495\textwidth}
    \includegraphics[width=\linewidth]{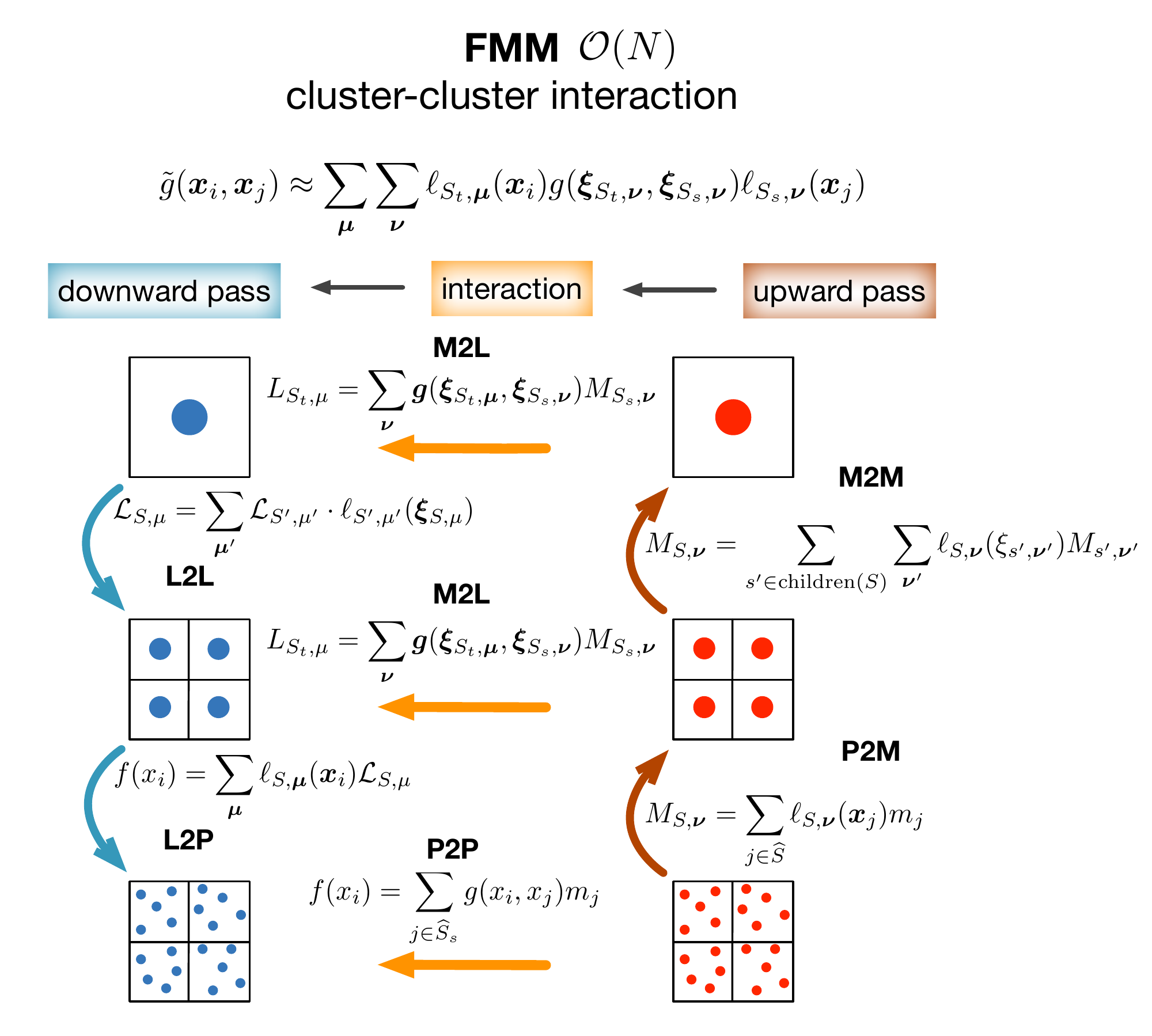}
    \caption{\label{fig:fmm_schematic_treecode_vs_fmm_b}FMM}
\end{subfigure}
\caption{A schematic comparison of treecode and FMM. The blue dots and red dots indicate the target and source particles, respectively.}
\label{fig:fmm_schematic_treecode_vs_fmm}
\end{figure}

The pseudocodes for the algorithms are presented with the following global variables:
\begin{center}
\begin{tabular}{ll}
	$\bm{E}_i$ & electric field experienced by the $i$-th particle,\\
	$\bm{B}_i$ & magnetic field experienced by the $i$-th particle,\\
	$\gamma_{S,\bm{\nu}}$ & effective Lorenz factor of a macro particle with index $\bm{\nu}$ in cluster $S$,\\ 
    $\bm{p}_{S,\bm{\nu}}$ & effective momentum of a macro particle with index $\bm{\nu}$ in cluster $S$,\\
    $\bm{E}_{S,\bm{\nu}}$ & total electric field acting on a macro particle with index $\bm{\nu}$ in cluster $S$,\\
    $\bm{B}_{S,\bm{\nu}}$ & total magnetic field acting on a macro particle with index $\bm{\nu}$ in cluster $S$.
\end{tabular}
\end{center}
The algorithms described above are implemented as a solver in the Julia programming language~\cite{bezanson2017julia}. The cluster tree constructed with \myref{Algorithm}{alg:fmm_subdivide_cluster} is implemented using a pointer-based data structure.

\begin{algorithm}\caption{Direction for the Split of Particle Cluster with Stretch}\label{alg:direction4split_stretching}
$S$: particle cluster\\
$\bm{s}$: stretch factor\\
\SetKwProg{Func}{Function}{}{end}
\Func{\normalfont{direction4split}$(S,\,\bm{s})$}{
    $(s_x,s_y,s_z)=\bm{s}$\\
    $[a_x,b_x]\times[a_y,b_y]\times[a_z,b_z]=\text{bbox}(S)$\\
    $k = \underset{i\in\{x,y,z\}}{\mathrm{argmax}}\,\,s_i(b_i - a_i)$\\
    \Return $k$
}
\end{algorithm}

\begin{algorithm}\caption{Subdivision of Particle Cluster}\label{alg:fmm_subdivide_cluster}
$S$: particle cluster\\
$\bm{s}$: stretch factor\\
$N_0$: maximum number of particles in the leaf node\\
\SetKwProg{Func}{Function}{}{end}
\Func{\normalfont{subdivide}$(S,\,\bm{s},\,N_0)$}{
  \eIf{$|S|>N_0$}{
    $k=\text{direction4split}(S,\bm{s})$ (\hyperref[alg:direction4split_stretching]{Algorithm~\ref*{alg:direction4split_stretching}})\\
    $k_{\text{split}}=$ $\lfloor |S|/2\rfloor$-th largest element of $\{k_i\mid i\in\widehat{S}\}$\\
    $S_1=\{\bm{x}_i \mid k_i\leq k_{\text{split}}, i\in\widehat{S}\}$\\
    $S_2=\{\bm{x}_i \mid k_i> k_{\text{split}}, i\in\widehat{S}\}$\\
    children(S) $= \{S_1,\,S_2\}$\\
    subdivide($S_1,\,\bm{s},\,N_0$)\\
    subdivide($S_2,\,\bm{s},\,N_0$)\\
  }{
    children(S) $= \emptyset$
  }
}
\end{algorithm}

\begin{algorithm}\caption{P2M}\label{alg:P2M}
$S$: particle cluster\\
\SetKwProg{Func}{Function}{}{end}
\Func{\normalfont{P2M}$(S)$}{
    $\gamma_{S,\bm{\nu}}=\sum_{j\in \widehat{S}}\ell_{S,\bm{\nu}}(\bm{x}_j)\gamma_{j}$ \\
    $\bm{p}_{S,\bm{\nu}}=\sum_{j\in \widehat{S}}\ell_{S,\bm{\nu}}(\bm{x}_j)\bm{p}_{j}$
}
\end{algorithm}

\begin{algorithm}\caption{M2M}\label{alg:M2M}
$S$: parent particle cluster\\
$S'$: child particle cluster\\
\SetKwProg{Func}{Function}{}{end}
\Func{\normalfont{M2M}$(S,S')$}{
    $\gamma_{S,\bm{\nu}}=\sum_{\bm{\nu}'}\ell_{S,\bm{\nu}}(\xi_{S',\bm{\nu}'})\gamma_{S',\bm{\nu}'}$\\
    $\bm{p}_{S,\bm{\nu}}=\sum_{\bm{\nu}'}\ell_{S,\bm{\nu}}(\xi_{S',\bm{\nu}'})\gamma_{S',\bm{\nu}'}$
}
\end{algorithm}

\begin{algorithm}\caption{Upward Pass}\label{alg:upwardpass}
$S$: particle cluster\\
\SetKwProg{Func}{Function}{}{end}
\Func{\normalfont{upwardpass}$(S)$}{
  \eIf{\normalfont{children}$(S) == \emptyset$}{
    $\text{P2M}(S)$ (\myref{Algorithm}{alg:P2M})\\
  }{
    \For{$s \in \text{\normalfont{children}}(S)$}{
       $\text{upwardpass}(s)$\\
    }
    \For{$s \in \text{\normalfont{children}}(S)$}{
       $\text{M2M}(S,s)$ (\myref{Algorithm}{alg:M2M})\\
    }
  }
}
\end{algorithm}

\begin{algorithm}\caption{Stretched Diameter of Cluster}\label{alg:stretched_diam}
$S$: particle cluster\\
$\bm{s}$: stretch factor\\
\SetKwProg{Func}{Function}{}{end}
\Func{\normalfont{diam}$(S, \bm{s})$}{
  $[a_x,b_x]\times[a_y,b_y]\times[a_z,b_z]=\text{bbox}(S)$\\
  \Return $\Vert\bm{s}\circ(\bm{a} - \bm{b})/2\Vert_2$
}
\end{algorithm}

\begin{algorithm}\caption{Stretched Distance between two Clusters}\label{alg:stretched_dist}
$S_1$: particle cluster 1\\
$S_2$: particle cluster 2\\
$\bm{s}$: stretch factor\\
\SetKwProg{Func}{Function}{}{end}
\Func{\normalfont{dist}$(S_1, S_2, \bm{s})$}{
  $[a_{1,x},b_{1,x}]\times[a_{1,y},b_{1,y}]\times[a_{1,z},b_{1,z}]=\text{bbox}(S_1)$\\
   $[a_{2,x},b_{2,x}]\times[a_{2,y},b_{2,y}]\times[a_{2,z},b_{2,z}]=\text{bbox}(S_2)$\\
  $\bm{c}_1 = (\bm{a}_1 + \bm{b}_1)/2$\\
  $\bm{c}_2 = (\bm{a}_2 + \bm{b}_2)/2$\\
  \Return $\Vert\bm{s}\circ(\bm{c}_1-\bm{c}_2)\Vert_2$
}
\end{algorithm}

\begin{algorithm}\caption{Stretched Admissibility Condition for Cluster-Cluster Interaction}\label{alg:stretched_admissible_CC}
$S_1$: particle cluster 1\\
$S_2$: particle cluster 2\\
$\bm{s}$: stretch factor\\
$\eta$: admissibility parameter\\
\SetKwProg{Func}{Function}{}{end}
\Func{\normalfont{admissible}$(S_{1},\,S_{2},\bm{s},\eta)$}{
   $r_1 = \text{diam}(S_1, \bm{s})$ (\myref{Algorithm}{alg:stretched_diam})\\
   $r_2 = \text{diam}(S_2, \bm{s})$ (\myref{Algorithm}{alg:stretched_diam})\\
   $d = \text{dist}(S_1,S_2,\bm{s})$ (\myref{Algorithm}{alg:stretched_dist})\\
   \Return $\max(r_1, r_2)/d < \eta$
}
\end{algorithm}

\begin{algorithm}\caption{M2L}\label{alg:M2L}
$S_t$: target particle cluster\\
$S_s$: source particle cluster\\
\SetKwProg{Func}{Function}{}{end}
\Func{\normalfont{M2L}$(S_t, S_s)$}{
  $\bm{E}_{S_t,\bm{\mu}} = \sum_{\bm{\nu}}\gamma_{S_s,\bm{\nu}}\cdot\bm{g}(\xi_{S_t,\bm{\mu}}, \xi_{S_s,\bm{\nu}})$\\
  $\bm{B}_{S_t,\bm{\mu}} = \sum_{\bm{\nu}}\bm{p}_{S_s,\bm{\nu}}\times\bm{g}(\xi_{S_t,\bm{\mu}}, \xi_{S_s,\bm{\nu}})$
}
\end{algorithm}

\begin{algorithm}\caption{P2P}\label{alg:P2P}
$S_t$: target particle cluster\\
$S_s$: source particle cluster\\
\SetKwProg{Func}{Function}{}{end}
\Func{\normalfont{P2P}$(S_t, S_s)$}{
  \For{$i\in \hat{S_t}$}{
    \For{$j\in \hat{S_s}$}{
      $\bm{E}_i = \bm{E}_i + \gamma_j \bm{g}(\bm{x}_i,\bm{x}_j,\bm{p}_j)$\\
      $\bm{B}_i = \bm{B}_i + \bm{p}_j \times\bm{g}(\bm{x}_i,\bm{x}_j,\bm{p}_j)$
    }
  }
}
\end{algorithm}

\begin{algorithm}\caption{Cluster-Cluster Interaction by Dual Tree Traversal}\label{alg:dualtraverse}
$S_t$: target particle cluster\\
$S_s$: source particle cluster\\
$\bm{s}$: stretch factor\\
$\eta$: admissibility parameter\\
\SetKwProg{Func}{Function}{}{end}
\Func{\normalfont{dualtraverseinteract}$(S_t, S_s, \bm{s}, \eta)$}{
  \eIf{\normalfont{children}$(S_t) == \emptyset$ $\land$  \normalfont{children}$(S_s) == \emptyset$}{
    $\text{P2P}(S_t, S_s)$ (\myref{Algorithm}{alg:P2P})\\
  }{
    $\text{isAdmissible} = \text{admissible}(S_{t},\,S_{s},\bm{s},\eta)$ (\myref{Algorithm}{alg:stretched_admissible_CC})\\
    \uIf{\normalfont{isAdmissible}}{
      \normalfont{M2L}$(S_{t},\,S_{s})$ (\myref{Algorithm}{alg:M2L})
    }\uElseIf{\normalfont{children}$(S_t) == \emptyset$}{
      \For{$s\in \text{\normalfont{children}}(S_s)$}{
        $\text{dualtraverseinteract}(S_t, s, \bm{s}, \eta)$
      }
    }\uElseIf{\normalfont{children}$(S_s) == \emptyset$}{
      \For{$t\in \text{\normalfont{children}}(S_t)$}{
        $\text{dualtraverseinteract}(t, S_s, \bm{s}, \eta)$
      }
    }\Else{
      \eIf{\normalfont{diam}$(S_t,\bm{s})$ $>$ \normalfont{diam}$(S_s,\bm{s})$}{
        \For{$t\in \text{\normalfont{children}}(S_t)$}{
          $\text{dualtraverseinteract}(t, S_s, \bm{s}, \eta)$
        }
      }{
        \For{$s\in \text{\normalfont{children}}(S_s)$}{
          $\text{dualtraverseinteract}(S_t, s, \bm{s}, \eta)$
        }
      }
    } 
  }
}
\end{algorithm}

\begin{algorithm}\caption{FMM with Stretch}\label{alg:FMM}
$S$: particle cluster\\
$\bm{s}$: stretch factor\\
$N_0$: maximum number of particles in the leaf
node\\
$\eta$: admissibility parameter\\
\SetKwProg{Func}{Function}{}{end}
\Func{\normalfont{FMM}$(S,\,\bm{s},\,N_{0},\,\eta)$}{
    \normalfont{subdivide}$(S,\,\bm{s},\,N_{0})$ (\myref{Algorithm}{alg:fmm_subdivide_cluster})\\
    $\bm{E}_i=\bm{0}, \bm{B}_i=\bm{0}\quad \forall i\in \widehat{S}$ \\
    $\text{upwardpass}(S)$ (\myref{Algorithm}{alg:upwardpass})\\ 
    $\text{dualtraverseinteract}(S,S,\bm{s},\eta)$ (\myref{Algorithm}{alg:dualtraverse})\\
    $\text{downwardpass}(S)$ (\myref{Algorithm}{alg:downwardpass})\\
}
\end{algorithm}

\begin{algorithm}\caption{L2L}\label{alg:L2L}
$S$: parent particle cluster\\
$S'$: child particle cluster\\
\SetKwProg{Func}{Function}{}{end}
\Func{\normalfont{L2L}$(S',S)$}{
    $\bm{E}_{S',\bm{\mu}'}=\bm{E}_{S',\bm{\mu}'}+\sum_{\bm{\mu}}\ell_{S,\bm{\mu}}(\xi_{S',\bm{\mu}'})\bm{E}_{S,\bm{\mu}}$\\
    $\bm{B}_{S',\bm{\mu}'}=\bm{B}_{S',\bm{\mu}'}+\sum_{\bm{\mu}}\ell_{S,\bm{\mu}}(\xi_{S',\bm{\mu}'})\bm{B}_{S,\bm{\mu}}$
}
\end{algorithm}

\begin{algorithm}\caption{L2P}\label{alg:L2P}
$S$: particle cluster\\
\SetKwProg{Func}{Function}{}{end}
\Func{\normalfont{L2P}$(S)$}{
  \For{$i\in \hat{S}$}{
    $\bm{E}_i = \bm{E}_i + \sum_{\bm{\mu}}\ell_{S,\bm{\mu}}(\bm{x}_i)\bm{E}_{S,\bm{\mu}}$\\
    $\bm{B}_i = \bm{B}_i + \sum_{\bm{\mu}}\ell_{S,\bm{\mu}}(\bm{x}_i)\bm{B}_{S,\bm{\mu}}$
  }
}
\end{algorithm}

\begin{algorithm}\caption{Downward Pass}\label{alg:downwardpass}
$S$: particle cluster\\
\SetKwProg{Func}{Function}{}{end}
\Func{\normalfont{downwardpass}$(S)$}{
  \eIf{\normalfont{children}$(S) == \emptyset$}{
    $\text{L2P}(S)$ (\myref{Algorithm}{alg:L2P})\\
  }{
    \For{$s \in \text{\normalfont{children}}(S)$}{
       $\text{L2L}(s,S)$ (\myref{Algorithm}{alg:L2L})\\
    }
    \For{$s \in \text{\normalfont{children}}(S)$}{
       $\text{downwardpass}(s)$\\
    }
  }
}
\end{algorithm}

\section{Results}\label{sec:fmm_result}
To understand the performance of the proposed FMM, we first demonstrate a plot of elapsed time against the error for the simulations with different FMM parameters in~\myref{Figure}{fig:fmm_etime_vs_error}. In each simulation, $1.28\times 10^6$ particles are randomly uniformly distributed in the unit cube $[0,1]^3$ and each particle has the same momentum $\bm{p}=(0,0,p_0)$ with $p_0=(\gamma^2-1)^{1/2}$ and $\gamma=50$. The measured error is the maximal relative error in the electrical and magnetic fields,
\begin{equation}\label{eq:fmm_error}
    \text{error}:=
      \max_{\bm{f}\in\{\bm{E},\bm{B}\}}
      \biggl(\sum^{N}_{i=1} \Vert \bm{f}^{t}_{i}-\bm{f}^{b}_{i} \Vert^{2}_{2} 
      / \sum^{N}_{i=i}\Vert \bm{f}^{b}_{i} \Vert^{2}_{2}\biggr)^{1/2},
\end{equation}
where $N$ is the number of particles in the system. The space-charge fields $\bm{f}^{t}_{i}$ and $\bm{f}^{b}_{i}$ experienced by the $i$-th particle are computed by FMM and a brute-force method (~\myref{Algorithm}{alg:P2P} with $S_t=S_s=S$ and $i \neq j$), respectively. We can observe that a smaller admissibility parameter $\eta$~\eqref{eq:fmm_stretched_admc} leads to higher accuracy (smaller error) but costs more elapsed time. This is because fewer M2Ls in the coarse level are performed and each non-admissible pair of clusters in the coarse level can result in many M2Ls in the finer level or P2Ps in the leaf level. 

The usage of a higher interpolation degree $n$ leads to a result with higher accuracy and higher elapsed time because more macro particles are used in the calculation of M2L. In \myref{Figure}{fig:fmm_performance}, the elapsed time of FMM against the number of particles $N$ is presented. We can see that our FMM approaches the theoretical complexity $\mathcal{O}(N)$ as the number of particles $N$ becomes big enough (\myref{Figure}{fig:fmm_performance_b}).   
\begin{figure}[h]
    \centering
    \includegraphics[width=0.65\linewidth]{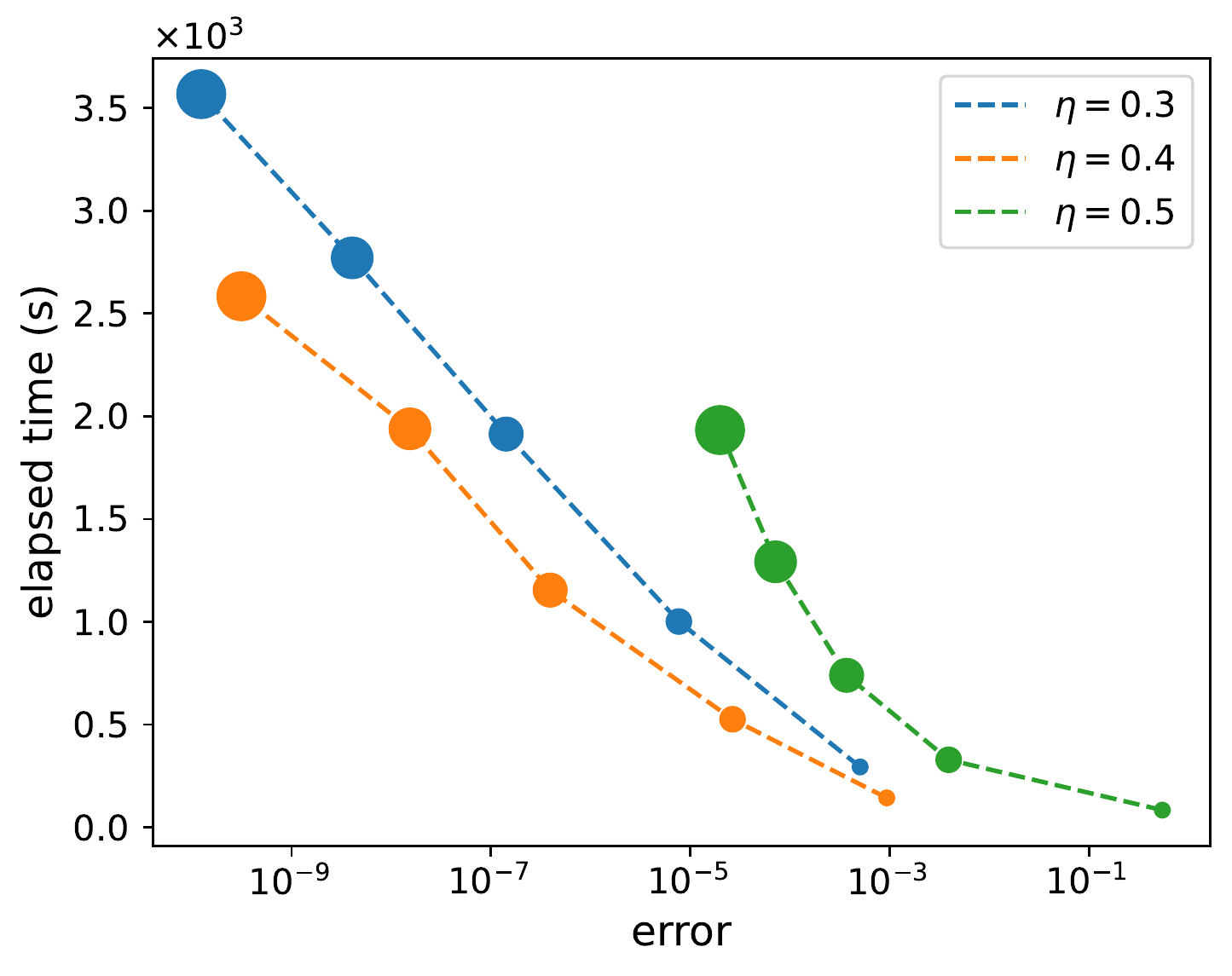}
    \caption{A plot of elapsed time against the error~\eqref{eq:fmm_error} for the proposed FMM. Each line represents the result computed with an admissibility parameter $\eta=0.3,0.4,0.5$. Each point in a line represents a simulation with an interpolation degree $n=2,4,6,8,10$ and the maximum number of particles in the leaf cluster $N_0=(n+1)^3$. The point size is associated with the value of $n$; data with bigger $n$ is expressed with bigger point size.}
    \label{fig:fmm_etime_vs_error}
\end{figure}
\begin{figure}[h]
\centering
\begin{subfigure}[b]{0.49\textwidth}
    \includegraphics[width=\linewidth]{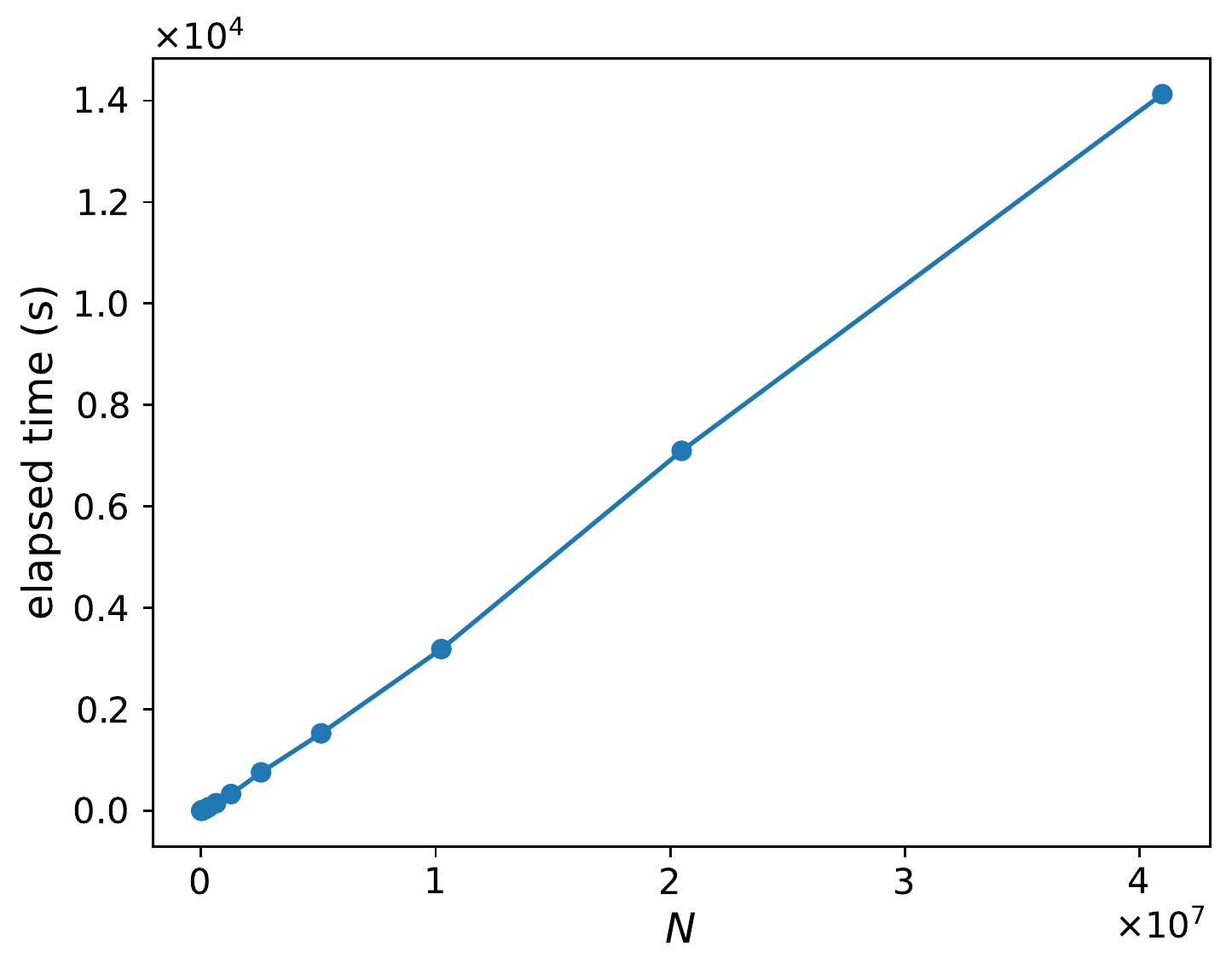}
    \caption{\label{fig:fmm_performance_a}elapsed time}
\end{subfigure}
\begin{subfigure}[b]{0.49\textwidth}
    \includegraphics[width=\linewidth]{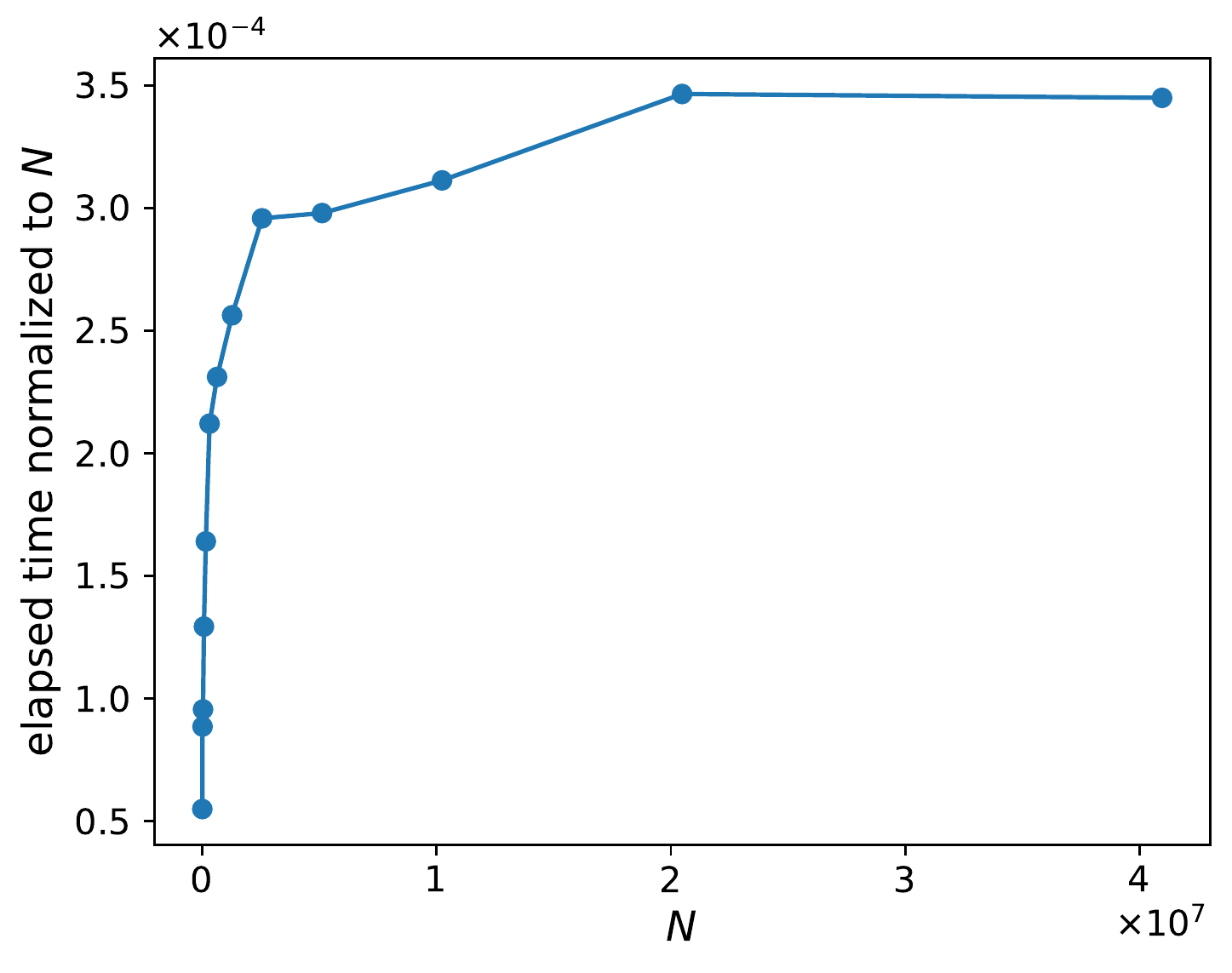}
    \caption{\label{fig:fmm_performance_b}normalized elapsed time}
\end{subfigure}
\caption{The performance of the FMM method. \myref{Figure}{fig:fmm_performance_a} shows the elapsed time used by FMM to evaluate the space-charge field of increasing numbers of particles $N$ with interpolation degree $n=4$, the maximum number of particles in the leaf cluster $N_0=(n+1)^3$ and admissibility parameter $\eta=0.5$. \myref{Figure}{fig:fmm_performance_b} shows the elapsed time normalized to $N$.}
\label{fig:fmm_performance}
\end{figure}

We also perform code profiling on our solver and demonstrate the cumulative elapsed time of the six FMM kernels in~\myref{Figure}{fig:fmm_routines_profiling}. One can observe that the total elapsed time is mostly dominated by P2P and M2L; this indicates that the routines of these two kernels will be the focus when any further optimizations for the solver are considered. Besides, one can also observe a sudden jump in the value of the elapsed times for P2P and M2L at a specific number of particles $N$. To understand this phenomenon, we consider a case where the total particle number is equal to a transition value $N=N^{\kappa}_{t}:=2^{\kappa}N_0$ with $\kappa$ the depth of the cluster tree. If each leaf cluster interacts with at most a constant number of clusters via P2P, the total number of operation counts to perform P2P can be written as
\begin{equation}\label{eq:fmm_performance_model_p2p_1}
    W_{\text{P2P}}(N)=\text{const}\cdot2^{\kappa}\cdot N^{2}_{0}\quad\text{for }N=N^{\kappa}_t.
\end{equation}
When the number of particles $N$ slightly increases with $\delta N\to 0$ so that $N > N^{\kappa}_t$, the number of particles in each leaf cluster $N_{\text{f}}$ will slightly increase with $\delta N_{\text{f}}\to 0$ so that $N_{\text{f}} > N_0$. In this case, each leaf cluster will be subdivided into two clusters and the cluster tree will gain one more level $\kappa+1$. Therefore, the number of leaf clusters will increase from $2^{\kappa}$ to $2^{\kappa+1}$ and the value of $N_{\text{f}}$ reduces from $N_0$ to $N_0/2$. Thus, the total number of operation counts for P2P can be written as
\begin{equation}\label{eq:fmm_performance_model_p2p_2}
    W_{\text{P2P}}(N)=\text{const}\cdot2^{\kappa+1}\cdot (\tfrac{N_0}{2} + \delta N_{\text{f}})^2 \quad\text{for }N^{\kappa+1}_t \geq N>N^{\kappa}_t,
\end{equation}
with $N_0/2 \geq \delta N_{\text{f}}(N,N_0) > 0 $. The ratio between \eqref{eq:fmm_performance_model_p2p_1} and \eqref{eq:fmm_performance_model_p2p_2} for different $\delta N_{\text{f}}$ is
\begin{equation*}
\dfrac{W_{\text{P2P}}(N^{\kappa}_t+\delta N)}{W_{\text{P2P}}(N^{\kappa}_t)}
=
\begin{cases}
    1\quad\text{for }\delta N_{\text{f}} = 0,\\
    \frac{1}{2}\quad\text{for }\delta N_{\text{f}}\to 0,\\
    2\quad\text{for } \delta N_{\text{f}}=\frac{N_0}{2}.
\end{cases}
\end{equation*}
Together with~\eqref{eq:fmm_performance_model_p2p_2}, we can see that $W_{\text{P2P}}$ suddenly decreases to one half of $W_{\text{P2P}}(N^{\kappa}_t)$ as $N$ increases from $N=N^{\kappa}_{t}$ and then grows quadratically until it is two times bigger than $W_{\text{P2P}}(N^{\kappa}_t)$ at $N=N^{\kappa+1}_{t}$. This performance model can describe the trend of elapsed time for P2P. Likewise, we can also apply a similar analysis to M2L and write down the corresponding performance model as
\begin{equation}\label{eq:fmm_performance_model_m2l}
W_{\text{M2L}}(N)
=
\begin{cases}
\text{const}\cdot 2^{\kappa+1}\cdot (n+1)^6 \quad N=N^{\kappa+1}_t,\\
\text{const}\cdot 2^{\kappa+2}\cdot (n+1)^6 \quad N^{\kappa+1}_t \geq N>N^{\kappa}_t.
\end{cases}
\end{equation}
Here, we use the fact that a balanced cluster tree with the depth $l$ contains $2^{l+1}$ total clusters and the assumption that each cluster interacts with at most a constant value of clusters through M2L. \myref{Equation}{eq:fmm_performance_model_m2l} shows $W_{\text{M2L}}$ suddenly increases to two times of $W_{\text{M2L}}(N^{\kappa}_t)$ as $N$ slightly increases with $\delta N\to 0$ from $N=N^{\kappa}_t$; and then it remains constant whenever $N^{\kappa+1}_t \geq N>N^{\kappa}_t$. This performance model can successfully explain the behavior of the elapsed time for M2L. 

\begin{figure}[H]
    \centering
    \includegraphics[width=0.6\linewidth]{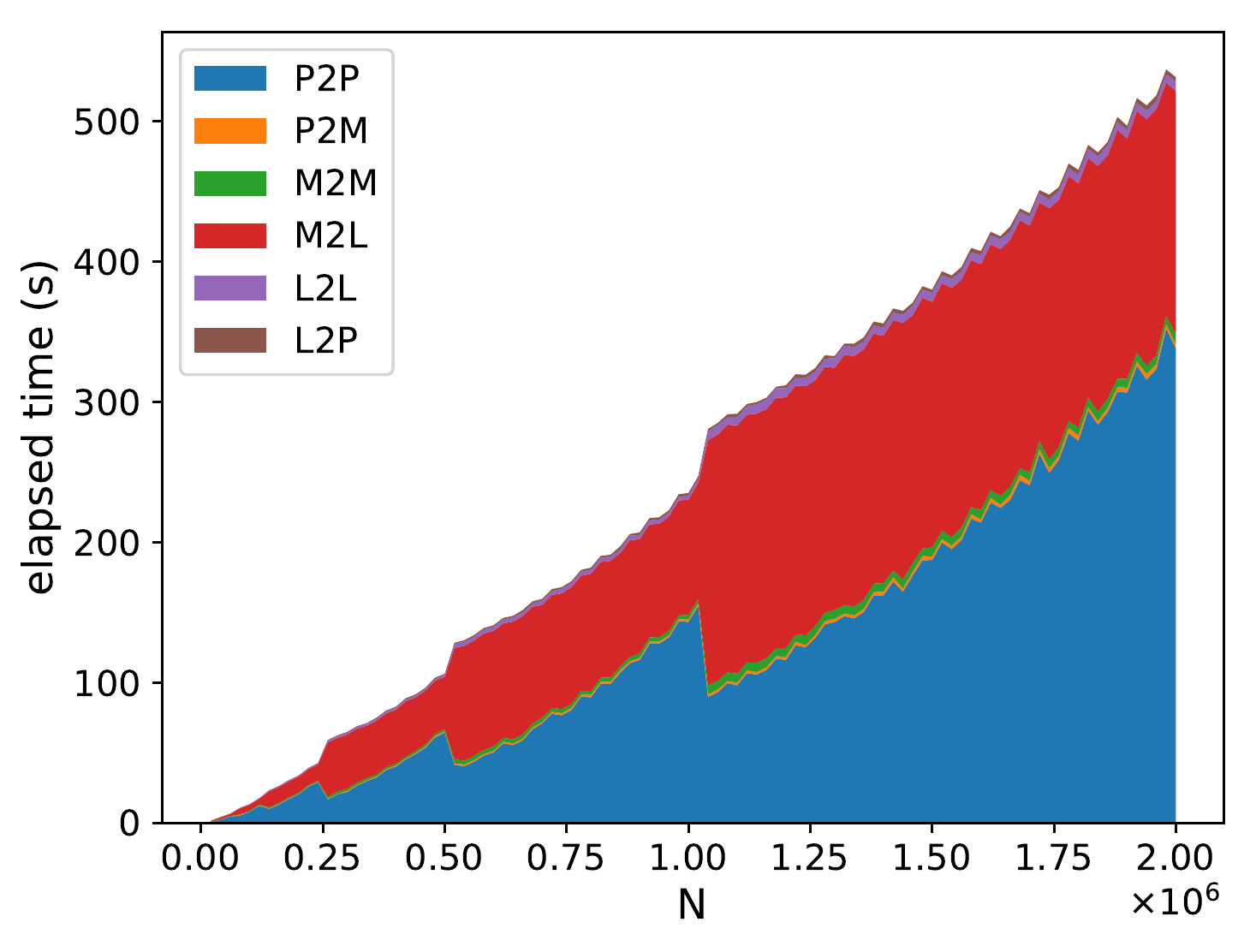}
    \caption{Elapsed time of the six FMM kernels against the increasing number of particles $N$ with interpolation degree $n=4$, the maximum number of particles in the leaf cluster $N_0=(n+1)^3$ and admissibility parameter $\eta=0.5$.}
    \label{fig:fmm_routines_profiling}
\end{figure}

\section{GPU Parallelization}
As illustrated in \myref{Algorithm}{alg:dualtraverse}, our FMM is based on the dual tree traversal. The dual tree traversal could not be naively parallelized in data parallelism and might not benefit from GPUs. For one thing, the power of GPUs comes from executing multiple simple tasks through multiple threads in SIMD (single instruction, multiple data); for another thing, a single GPU core usually has weaker computing power than a single CPU core. Therefore, hybrid CPU-GPU approaches based on the creation of the interaction lists by CPU were investigated in some former works~\cite{wilson2021gpu-accelerated,liu2019efficient}. In this approach, the CPU first performs a dual tree traversal to generate interaction lists; and then, the GPU handles the interaction of each pair of clusters in the interaction lists. In this study, we refer to the work proposed by Wilson~\emph{et al.}~\cite{wilson2021gpu-accelerated,wilson2021development} and discuss a GPU parallelization for our proposed FMM. The CPU-GPU execution of the proposed FMM can be summarized in the 10 steps listed in~\myref{Algorithm}{alg:fmmgpu_outline}. The H2D and D2H denote the data transfers of “host to device” and “device to host”, respectively. As shown in~\myref{Figure}{fig:fmm_routines_profiling}, the execution of FMM spends most of the time on the interaction phases (P2P and M2L). Although the parallelization of each FMM kernel is implemented in our application, we will only focus on the implementations of P2P and M2L (\emph{i.e.}, step 6 and step 8 in~\myref{Algorithm}{alg:fmmgpu_outline}) in later discussions.
{\LinesNumbered
\begin{algorithm}[H]\caption{The CPU-GPU Ecxecution of the Proposed Parallelized FMM}
\textcolor{red}{CPU}: generate particles information $\bm{x}_{i}$, $\bm{p}_{i}$ and allocate $\bm{E}_{i}$, $\bm{B}_{i}$ \\
\textcolor{red}{CPU}: create cluster tree $S$ with $\bm{x}_{i}$, $\bm{p}_{i}$\\
\textcolor{red}{H}2\textcolor{blue}{D}: copy $\bm{x}_{i}$, $\bm{p}_{i}$ and $S$ to device\\
\textcolor{blue}{GPU}: allocate $\gamma_{S,\bm{\nu}}$, $\bm{p}_{S,\bm{\nu}}$, $\bm{E}_{S,\bm{\nu}}$, $\bm{B}_{S,\bm{\nu}}$, $\bm{E}_{i}$, $\bm{B}_{i}$ \\
\textcolor{blue}{GPU}: perform upwardpass with $S$, $\bm{x}_{i}$, $\bm{p}_{i}$ to compute $\gamma_{S,\bm{\nu}}$, $\bm{p}_{S,\bm{\nu}}$ \\
\textcolor{red}{CPU}: perform dual tree traversal on $S$ to build interaction lists (ITLs)\\
\textcolor{red}{H}2\textcolor{blue}{D}: copy ITLs to device \\
\textcolor{blue}{GPU}: performe P2P, M2L with ITLs to compute $\bm{E}_{S,\bm{\nu}}$, $\bm{B}_{S,\bm{\nu}}$ \\
\textcolor{blue}{GPU}: perform downwardpass with $\bm{E}_{S,\bm{\nu}}$, $\bm{B}_{S,\bm{\nu}}$ to compute $\bm{E}_{i}$, $\bm{B}_{i}$ \\
\textcolor{blue}{D}2\textcolor{red}{H}: copy $\bm{E}_{i}$, $\bm{B}_{i}$ to host
\label{alg:fmmgpu_outline}
\end{algorithm}}

\section{Array-based Tree Data Structure}\label{sec:fmm_array_tree}
In the implementation, it might be straightforward to express the cluster tree with a pointer-based data structure; that is, each node object (particle cluster in our case) contains data fields and a pointer, and this pointer is used to allocate the objects of children nodes. One major disadvantage of using a pointer-based tree is that the node objects are not stored in contiguous locations in the memory, and this makes the data transfer between host and device difficult. Hence, it might be beneficial to consider an array-based tree in the GPU application. Following the approach in Ref.~\cite{hwu2011GPU_ch6}, we use multiple arrays to store node objects with multiple members, one array for one member. A member of a node object with index $i$ is located in the $i$-th element of the corresponding array. Besides, two additional arrays are respectively used to specify the parent index and children indices of nodes. Because our cluster tree is constructed through a k-d tree with cardinality-balanced subdivision of particles, it will be a balanced binary tree. Hence, we will narrow our following discussions to balanced binary tree. 

Although there may exist several possibilities, we adapt the breadth-first scheme to assign the node index of a tree. With this index assignment scheme, the nodes in the level $l$ have the indices $\{2^{l},\dots,2^{(l+1)-1}\}$; and similarly, a node with the index $i$ belongs to a level $\lfloor\log i/\log 2\rfloor$. Here, we define that the level of cluster tree starts from $0$ and the node index starts from $1$. The breadth-first scheme can ensure that the member data of nodes from the same level stays contiguously in an array. This data arrangement is cache-friendly for both the upward pass (P2M and M2M) and the downward pass (L2L and L2P) where the whole member data of nodes from the same level will be accessed for the calculation. Therefore, the parent index and the children pair of indices for a node with index $i$ are defined as
\begin{equation*}
    \text{iparent}(i)=
    \begin{cases}
    -1\quad &i=1, \\
    \left\lfloor\frac{i}{2}\right\rfloor\quad &i \neq 1,\\
    \end{cases}
    \quad
    \text{and}
    \quad
    \text{ichildren}(i)=
    \begin{cases}
    (-1,-1)\quad &\text{if leaf node},\\
    (2\cdot i,2\cdot i+1)\quad &\text{else.}
    \end{cases}
\end{equation*}
A schematic representation of our array-based tree is provided in~\myref{Figure}{fig:array_tree}.
Since our tree is balanced (due to the cardinality-balanced subdivision scheme), we can preallocate a fixed-size array by knowing that the total number of nodes is $2^{(\kappa+1)}-1$ with the tree-depth
\begin{equation*}
\kappa =
\begin{cases}
0                                 & N \leq N_0, \\
\lceil \log (N/N_0)/\log 2 \rceil & N > N_0.
\end{cases}
\end{equation*}
It is worth noting that the resulting tree might not be balanced if other space-subdivision schemes are adapted. In this case, one may preallocate a big enough array that each node contains the children nodes of a maximum possible number. However, this causes large memory of unused nodes and leads to poor load-balancing across multiple ranks when MPI parallelization is considered~\cite{wilson2022private_communication}. One possible way to work around this issue is first creating a pointer-based tree, and an array-based tree can be allocated based on the information from that. This approach is adapted by some solvers, \emph{e.g.}, BaryTree~\cite{BaryTree}.
\begin{figure}[h]
    \centering
    \includegraphics[width=0.6\textwidth]{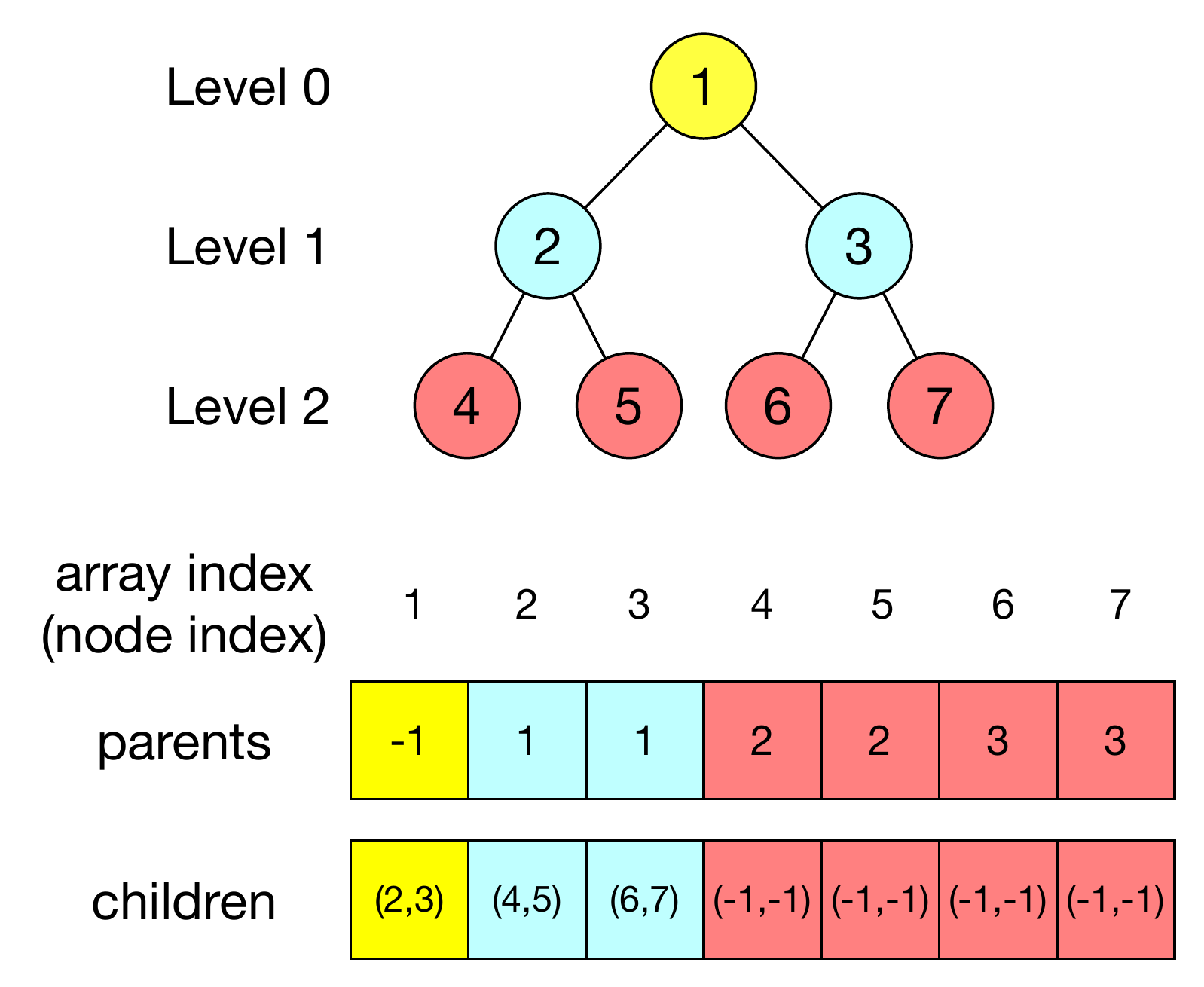}
    \caption{An array-representation of a balanced binary tree. The node index is assigned with the breadth-first scheme.}
    \label{fig:array_tree}
\end{figure}

\section{Parallelization of P2P and M2L Kernels}
Two lists of interaction pairs respectively for P2P and M2L, called interaction lists (ITLs), are generated by the execution of the dual tree traversal (\myref{Algorithm}{alg:dualtraversefillitl}) in the CPU and copied into GPU. The GPU kernels respectively of P2P and M2L are launched in a way that each interaction pair is handled by a thread block. As P2P and M2L are both similar to a direct summation algorithm, their GPU implementations are straightforward; one thread in the threads block handles the evaluation of the force-field of one target \mbox{micro/macro} particle (P2P/M2L). The GPU parallelization of P2P is illustrated in~\myref{Figure}{fig:fmmgpu_P2P_kernel}. In our implementation, we use an additional array to store the indices of all particles in the system and the indices of particles from a cluster will always stay in a contiguous memory block in the array during the subdivision (cf. \myref{Section}{appendx:treecode_clustertree_datastructure}). However, each member data of particles (\emph{e.g.,} positions and momenta) from a cluster accessed through this particle-indices array does not necessarily stay contiguous in its array (\myref{Figure}{fig:fmmgpu_P2P_kernel_a}). Thus, a member data of source particles accessed by threads is non-contiguously distributed in an array. This can slow down the application because the data access is not cache-friendly and requires frequent access from the global memory. One way to remedy this problem is using the shared memory (\myref{Figure}{fig:fmmgpu_cuda_gpu}) provided in CUDA-capable GPUs: we first load each member data of particles from a source cluster to shared memory so that the data can be accessed much faster by threads (\myref{Figure}{fig:fmmgpu_P2P_kernel_b}). For one thing, each member data of source particles stays in a contiguous block in the shared memory. For another, the shared memory is on-chip memory and has much lower latency than the global memory.
\begin{figure}[h]
    \centering
    \includegraphics[width=0.6\linewidth]{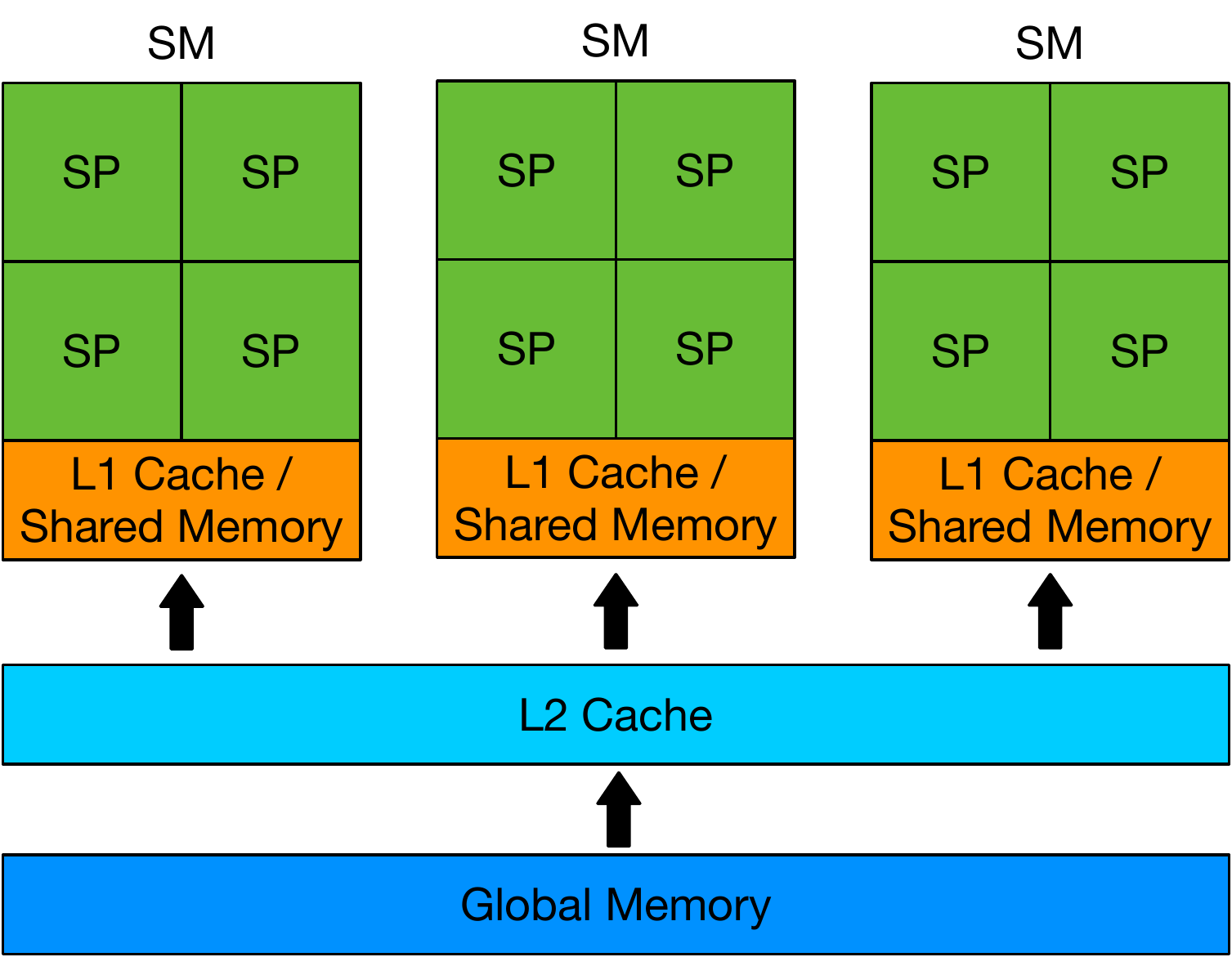}
    \caption{Memory hierarchy of CUDA-capable GPUs. SP and SM denote streaming processor and streaming multiprocessor, respectively.}
    \label{fig:fmmgpu_cuda_gpu}
\end{figure}
\begin{figure}[h]
\centering
\begin{subfigure}[b]{0.40\linewidth}
    \includegraphics[width=\linewidth]{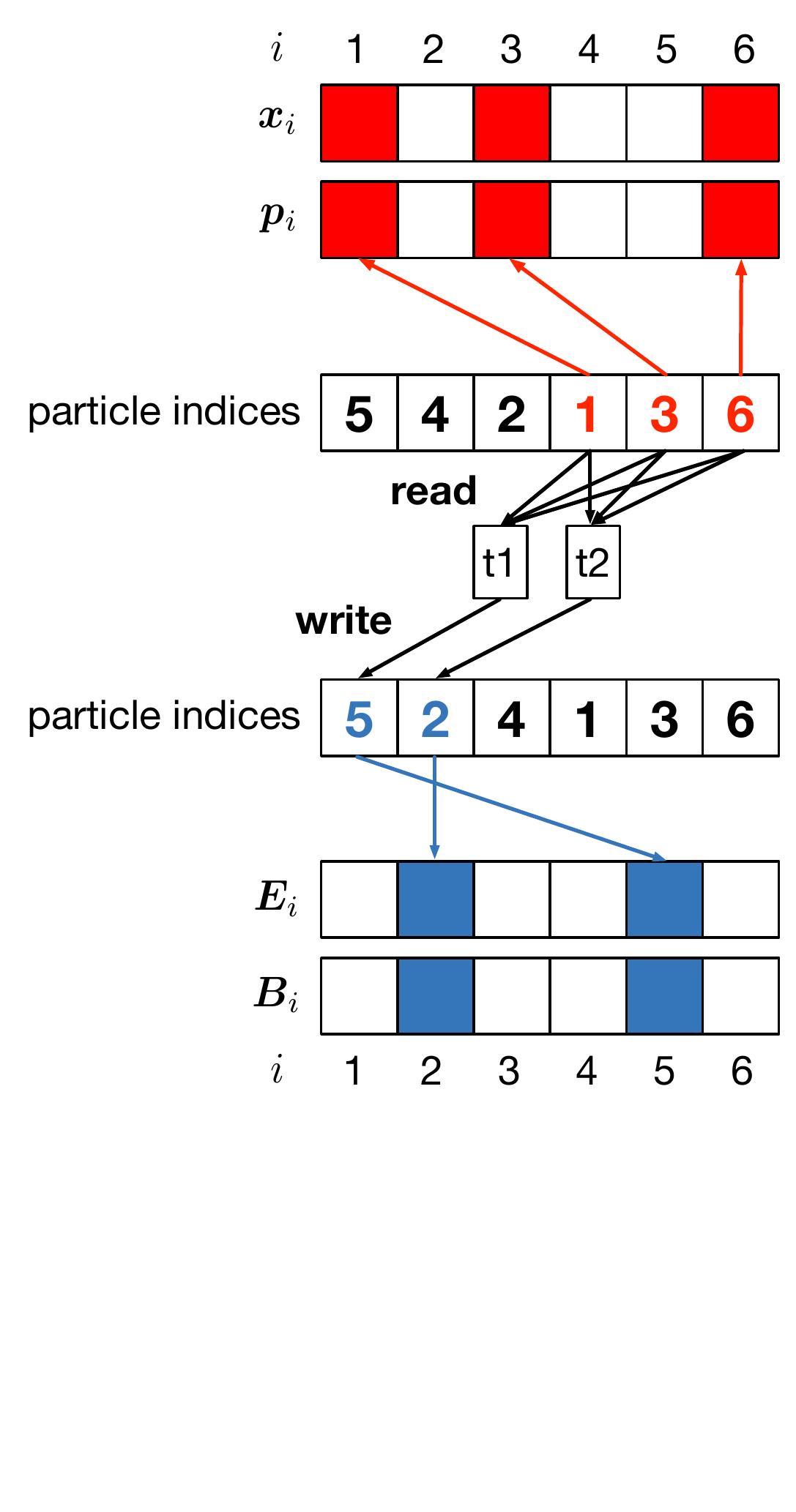}
    \caption{\label{fig:fmmgpu_P2P_kernel_a}P2P without shared memory}
\end{subfigure}
\hfill%
\begin{subfigure}[b]{0.40\linewidth}
    \includegraphics[width=\linewidth]{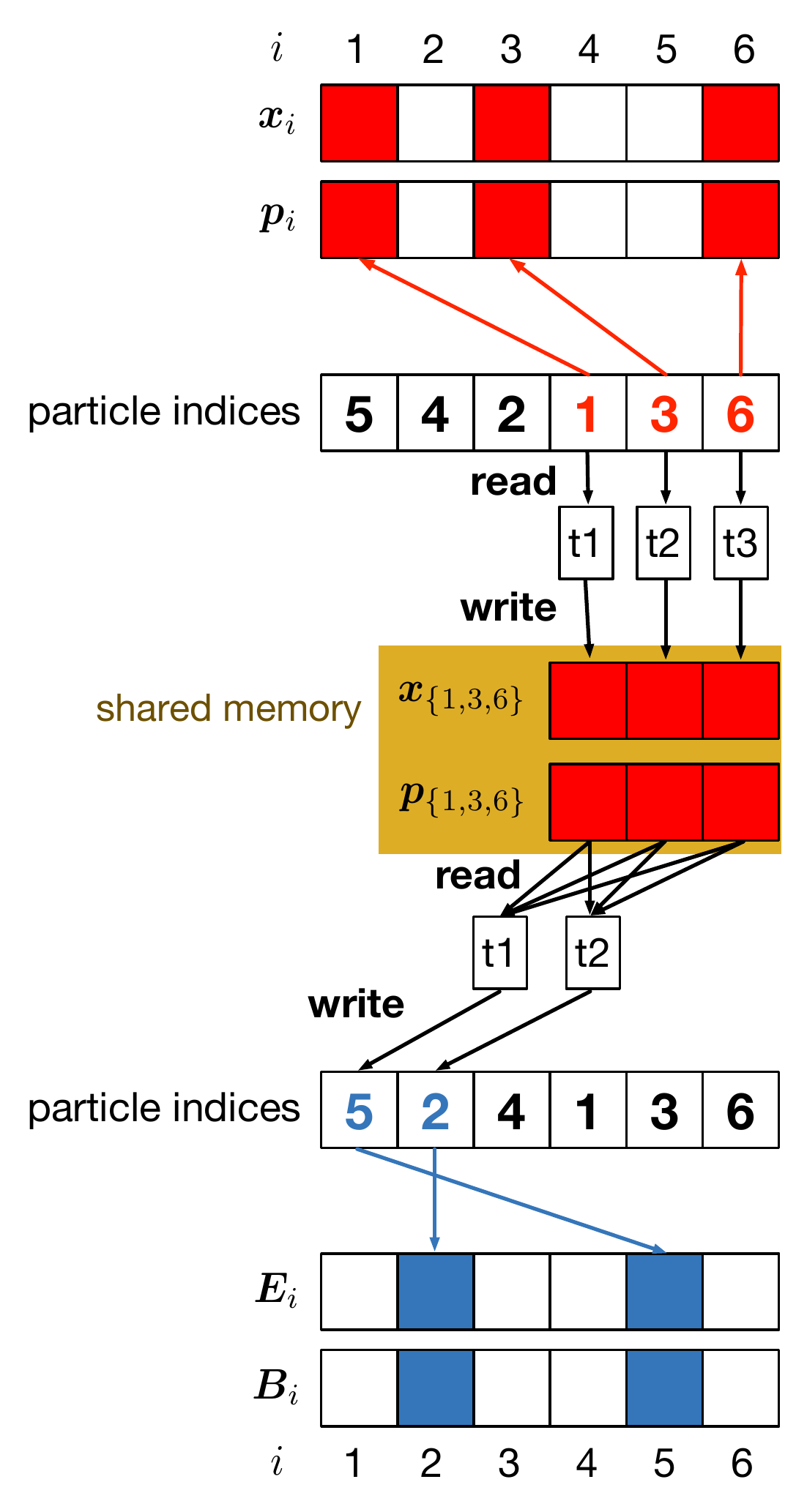}
    \caption{\label{fig:fmmgpu_P2P_kernel_b}P2P with shared memory}
\end{subfigure}
\caption{Implementations of P2P kernels (a) without shared memory and (b) with shared memory. The data associated with the target and source particle is colored in blue and red, respectively. The thread is denoted by a shorthand “t”.}
\label{fig:fmmgpu_P2P_kernel}
\end{figure}
For the M2L implementation, it is not necessary to apply the shared memory, because the member data of macro particles from a cluster is originally in a contiguous memory block. The data access is already cache-friendly so that the L1 cache in each streaming multiprocessor can be effectively used. A schematic of M2L implementation is provided in~\myref{Figure}{fig:fmmgpu_M2L_kernel}. Because the implementations of the other FMM kernels share large similarities with P2P or M2L, we will not go through the details of the implementations.   
\begin{figure}[H]
\centering
\includegraphics[width=0.3\linewidth]{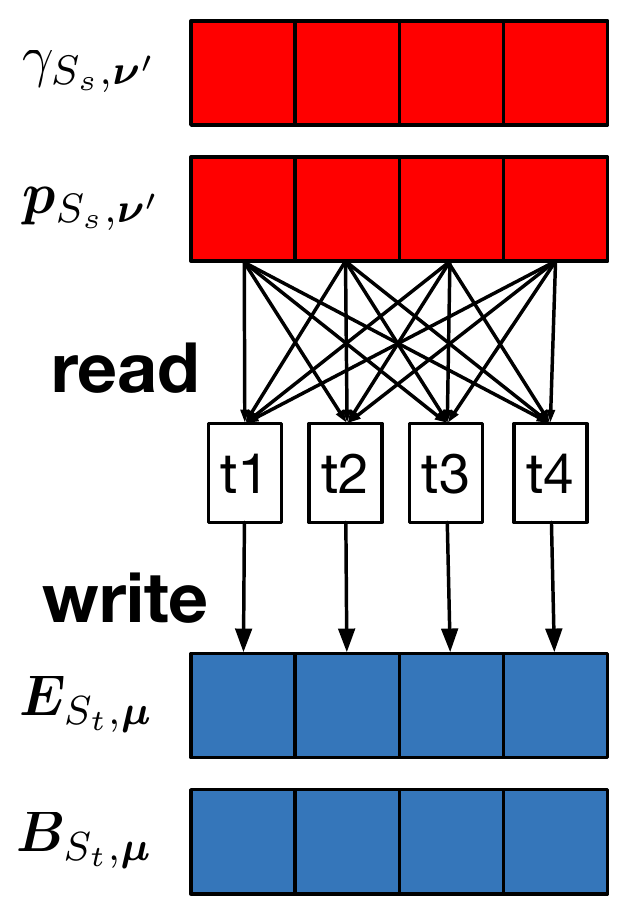}
\caption{An illustration of the implementation of the M2L kernel. The data associated with the macro particle of the target cluster and the source cluster are colored with blue and red, respectively. The thread is denoted by a shorthand “t”.}
\label{fig:fmmgpu_M2L_kernel}
\end{figure}

\begin{algorithm}\caption{Generation of Interaction Lists by Dual Tree Traversal}\label{alg:dualtraversefillitl}
$S_i$: particle cluster with node index $i$\\
$\bm{s}$: stretch factor\\
$\eta$: admissibility parameter\\
$\text{p2p\_itl}$: initially empty list of interaction for P2P (global scope)\\
$\text{m2l\_itl}$: initially empty list of interaction for M2L (global scope)\\
\SetKwProg{Func}{Function}{}{end}
\Func{\normalfont{dualtraversefillitl}$(i, j, \bm{s}, \eta)$}{
  \eIf{\normalfont{ichildren}$(i) == (-1,-1)$ $\land$  \normalfont{ichildren}$(j) == (-1,-1)$}{
    push $(i,j)$ to $\text{p2p\_itl}$\\
  }{
    $\text{isAdmissible} = \text{admissible}(S_{i},\,S_{j},\bm{s},\eta)$ (\myref{Algorithm}{alg:stretched_admissible_CC})\\
    \uIf{\normalfont{isAdmissible}}{
      push $(i,j)$ to $\text{m2l\_itl}$
    }\uElseIf{\normalfont{ichildren}$(i) == (-1,-1)$}{
      \For{$k \in \text{\normalfont{ichildren}}(j)$}{
        $\text{dualtraversefillitl}(i, k, \bm{s}, \eta)$
      }
    }\uElseIf{\normalfont{ichildren}$(j) == (-1,-1)$}{
      \For{$k\in \text{\normalfont{ichildren}}(i)$}{
        $\text{dualtraversefillitl}(k, j, \bm{s}, \eta)$
      }
    }\Else{
      \eIf{\normalfont{diam}$(S_i,\bm{s})$ $>$ \normalfont{diam}$(S_j,\bm{s})$}{
        \For{$k\in \text{\normalfont{ichildren}}(i)$}{
          $\text{dualtraversefillitl}(k, j, \bm{s}, \eta)$
        }
      }{
        \For{$k\in \text{\normalfont{ichildren}}(j)$}{
          $\text{dualtraversefillitl}(i, k, \bm{s}, \eta)$
        }
      }
    } 
  }
}
\end{algorithm}

\section{Race Conditions in P2P and M2L Kernels}
In CUDA applications, a GPU kernel can be launched with a grid of thread blocks and several thread blocks can be executed by streaming multiprocessors concurrently. In the execution of M2L or P2P, each pair of interaction is handled by one thread block and it is possible that several pairs of interaction with the same target index are handled by different thread blocks simultaneously. This can cause a race condition and produce an unexpected result because the corresponding memory data associated with a target cluster can be updated by the threads of different thread blocks at the same time (\myref{Figure}{fig:fmmgpu_race_condition_a}). One common remedy for the race conditions is using CUDA's atomic operations~\cite{cheng2014professional}, which locks a memory location so that only one exclusive thread is allowed to update the value each time. However, the atomic operations in CUDA only support some primitive types (\emph{e.g.}, \code{Int32} and \code{Float32}) and cannot be used in our implementation, because each three-dimensional vector in the physical system (\emph{e.g.}, position, momentum and vector field) is represented by a non-primitive and immutable type \code{SVector\{3,T\}}~\cite{staticarrays} with three elements of a parametric type \code{T}. Due to this immutability, we cannot apply atomic operations to change any elements of a \code{SVector\{3,T\}} object even though \code{T} is a primitive type (if so we can apply atomic operations to update each element of a \code{SVector\{3,T\}} object). Therefore, in our implementation, we divide pairs of interaction into groups such that each group only contains the pairs of interaction with the same target index; and during the kernel execution, each group will be handled by a thread block. To implement this, we first sort the pairs by the value of target index, which can be done efficiently with Quicksort. After that, we generate an additional array to indicate the start position of each group of pairs in the sorted ITL so that this array can be used to dispatch thread blocks to each group during the kernel execution (\myref{Figure}{fig:fmmgpu_race_condition_b}).
\begin{figure}[h]
\centering
\begin{subfigure}[b]{0.40\linewidth}
    \includegraphics[height=\textwidth]{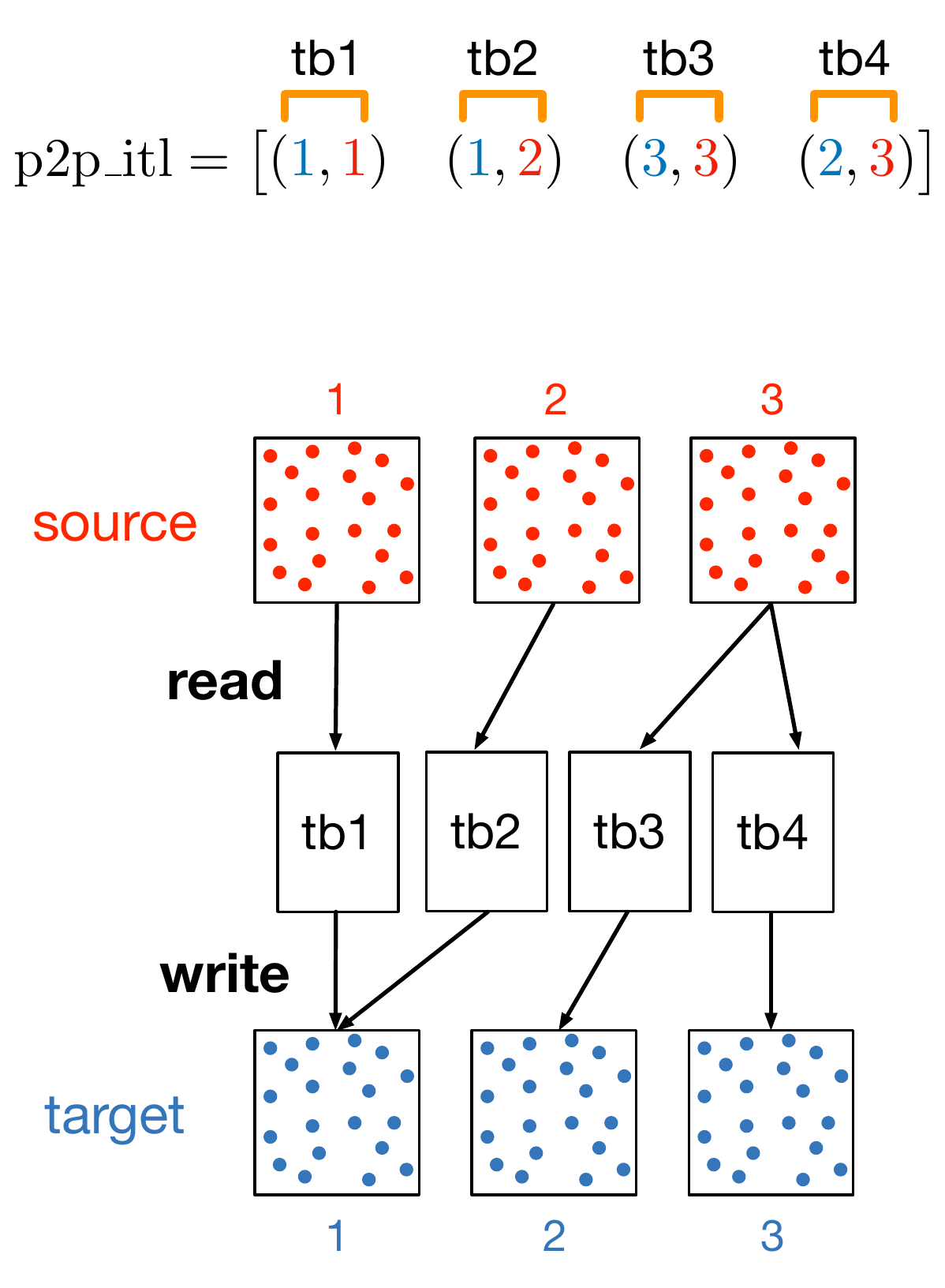}    \caption{\label{fig:fmmgpu_race_condition_a}original problem}
\end{subfigure}
\begin{subfigure}[b]{0.40\linewidth}
    \includegraphics[height=\textwidth]{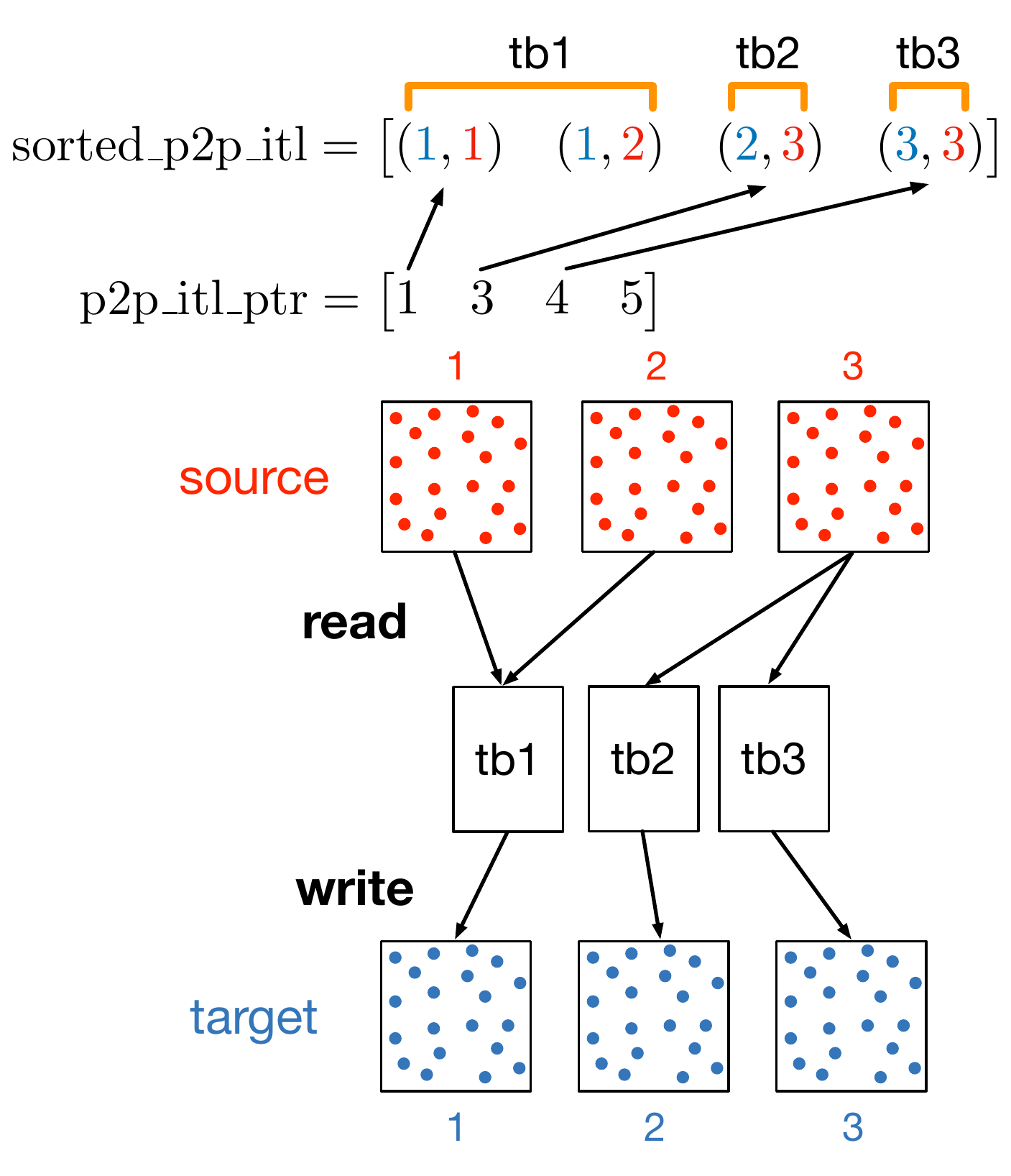}
    \caption{\label{fig:fmmgpu_race_condition_b}a solution}
\end{subfigure}
\caption{A schematic of the race condition problem in a naive parallelization of the P2P kernel (same for the M2L kernel) with ITL. \myref{Figure}{fig:fmmgpu_race_condition_a} illustrates the original problem. \myref{Figure}{fig:fmmgpu_race_condition_b} illustrates a solution by dividing the pairs of interaction into groups with the same target index. The thread block is denoted by a shorthand “tb”.}
\label{fig:fmmgpu_race_condition}
\end{figure}

\section{Performance}\label{sec:fmm_performance}
In this study, a package \href{https://github.com/ykkan/FMM4RBGPU.jl}{FMM4RBGPU.jl} is written in the Julia programming language~\cite{bezanson2017julia} with CUDA.jl~\cite{cudajl,besard2019effective}. This package provides CPU (serial) and GPU solvers for the FMM proposed in this study. The cluster tree used in this package is implemented with the array-based data structure discussed in~\myref{Section}{sec:fmm_array_tree}.

To understand the performance of our GPU parallelization, we consider a simulation with $N=2.56\times 10^7$ particles on different CPUs and GPUs as listed in~\myref{Table}{tb:fmmgpu_hardwards}. The elapsed times of the simulations (\emph{i.e.}, the execution of \myref{Algorithm}{alg:FMM}) are demonstrated in~\myref{Figure}{fig:fmmgpu_etime}. We can see that our GPU-based solver can achieve a speedup between $57$ and $197$ (relative to the result of a single CPU).

\begin{table}[ht]\renewcommand{\arraystretch}{1.2}\setlength{\tabcolsep}{40pt}
\centering
\begin{tabular}{ll}
\hline
\multicolumn{1}{c}{CPU}  & \multicolumn{1}{c}{GPU}   \\ 
\hline\hline
INTEL XEON E5-2640V4      & NVIDIA A100   \\
AMD EPYC 7402 &   NVIDIA V100    \\
INTEL XEON GOLD 5115   &  NVIDIA P100  \\
\hline
\end{tabular}
\caption{CPUs and GPUs used in the simulations for the performance benchmark.}
\label{tb:fmmgpu_hardwards}
\end{table}

\begin{figure}[h]
\centering
\begin{subfigure}[b]{0.495\linewidth}
    \includegraphics[width=\textwidth]{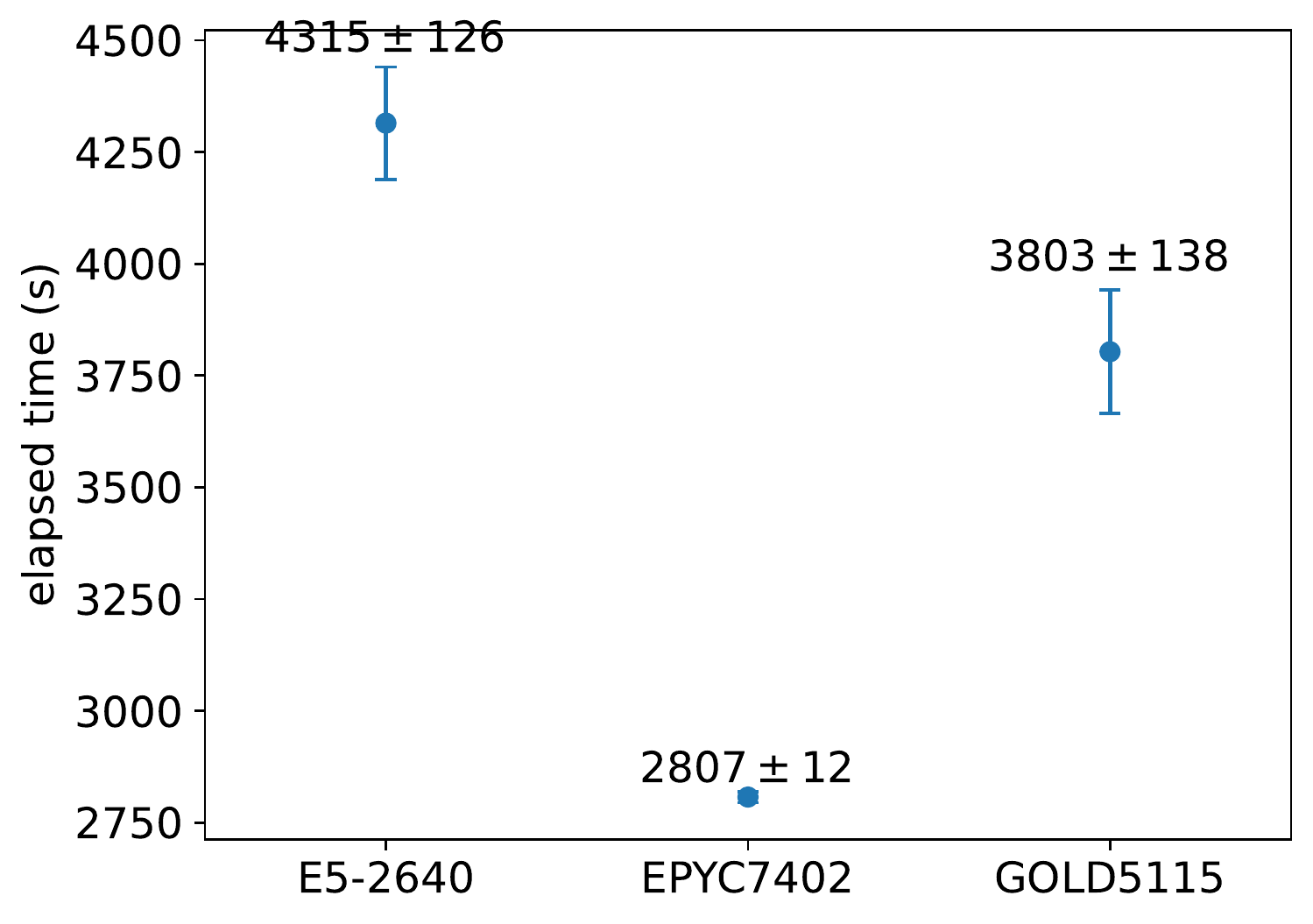}    
    \caption{\label{fig:fmmgpu_etime_a}CPU}
\end{subfigure}
\hfill%
\begin{subfigure}[b]{0.48\linewidth}
    \includegraphics[width=\textwidth]{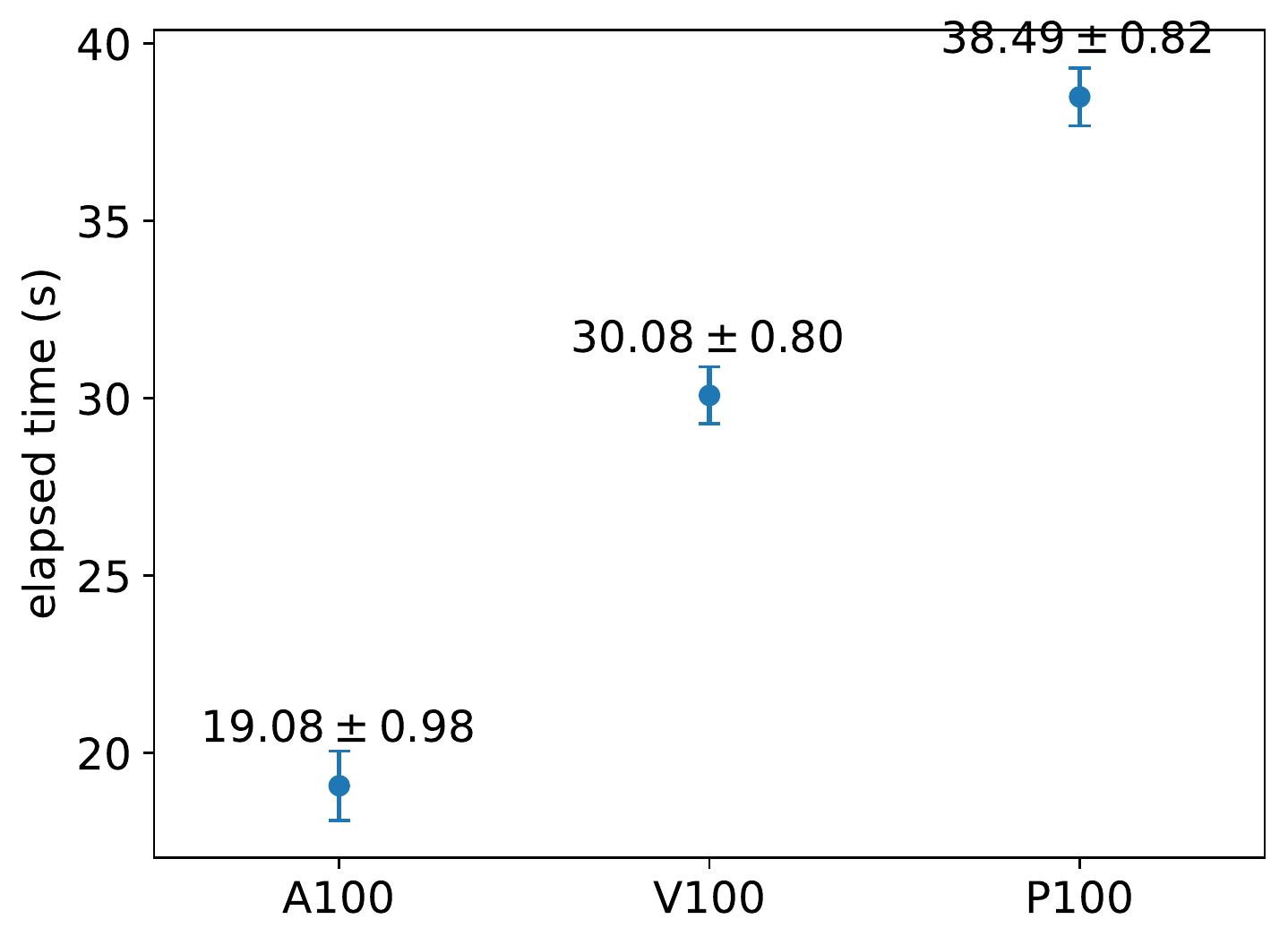}
    \caption{\label{fig:fmmgpu_etime_b}GPU}
\end{subfigure}
\caption{Elapsed times of simulations with $N=2.56\times 10^7$ particles on different (a) CPUs and (b) GPUs. The simulation is performed with $\eta=0.5$, $n=4$ and $N_0=(n+1)^3$. Each data point is the statistical result of 100 samples.}
\label{fig:fmmgpu_etime}
\end{figure}


\section{Summary}
In this study, we propose an interpolation-based FMM for the computation of the relativistic space-charge field. With our proposed modified admissibility condition, the FMM can be directly evaluated in the lab-frame without the need of a Lorentz transformation. We also consider a GPU parallelization for the proposed FMM. The pseudocode of the algorithms is provided and a corresponding package is developed in the Julia programming language. The proposed algorithms and package can be used to model the space-charge effect in the beam dynamics simulation of relativistic beams.

\appendix
\section{Definition of Cumulative Local Field}\label{appendix:fmm_proof_cumulative_local}
\begin{lemma}
Assume a target point $\bm{x}_i$ is contained in a sequence of clusters~${\{S^{l} \mid l=0,\dots,k \}}$ of each level $l$ with $S^{l+1}\subset S^{l}$ and $S^{0}=S$. The total force-field of the macro particles from this sequence of clusters transferred to this target point can be calculated by
\begin{equation*}
    f(\bm{x}_i)=
    \sum^{k}_{l=0}\sum_{\bm{\mu}}\ell_{S^{l},\bm{\mu}}(\bm{x}_i)L_{S^{l},\bm{\mu}} =\sum_{\bm{\mu}}\mathcal{L}_{S^{k},\bm{\mu}}\ell_{S^{k},\bm{\mu}}(\bm{x}_i),
\end{equation*}
where $\mathcal{L}_{S^{l}}$ is defined as 
\begin{equation*}
    \mathcal{L}_{S^{l},\bm{\mu}}:= L_{S^{l},\bm{\mu}}+ 
    \sum_{\bm{\mu}'}\mathcal{L}_{S^{l-1},\bm{\mu}'}\cdot\ell_{S^{l-1},\bm{\mu}'}(\bm{\xi}_{S^{l},\bm{\mu}})
    \quad\text{with}\quad
    \mathcal{L}_{S^{0},\bm{\mu}}:=L_{S^{0},\bm{\mu}}.
\end{equation*}
\end{lemma}
\begin{proof}
We prove this statement by mathematical induction. By the definition above, the statement automatically holds for the case $l=0$. We only need to prove the statement holds for the case $l=k$ provided that it is true for $l=k-1$. Thus, we have
\begin{align*}
    &\,\,\,\sum^{k}_{l=0}\sum_{\bm{\mu}}\ell_{S^{l},\bm{\mu}}(\bm{x}_i)L_{S^{l},\bm{\mu}} \\
    =&\,\,\,\sum_{\bm{\mu}}\ell_{S^{k},\bm{\mu}}(\bm{x}_i)L_{S^{k},\bm{\mu}} + \sum^{k-1}_{l=0}\sum_{\bm{\mu}}\ell_{S^{l},\bm{\mu}}(\bm{x}_i)L_{S^{l},\bm{\mu}}\\
    =&\,\,\,\sum_{\bm{\mu}}\ell_{S^{k},\bm{\mu}}(\bm{x}_i)L_{S^{k},\bm{\mu}} + \sum_{\bm{\mu}'}\mathcal{L}_{S^{k-1},\bm{\mu}'}\cdot\ell_{S^{k-1},\bm{\mu}'}(\bm{x}_i)\quad \text{(by assumption})\\
    \upeq{\eqref{eq:interpolation_of_polynormial}}&\,\,\, \sum_{\bm{\mu}}\ell_{S^{k},\bm{\mu}}(\bm{x}_i)L_{S^{k},\bm{\mu}} + \sum_{\bm{\mu}'}\mathcal{L}_{S^{k-1},\bm{\mu}'}\sum_{\bm{\mu}}\ell_{S^{k-1},\bm{\mu}'}(\bm{\xi}_{S^{k},\bm{\mu}})\ell_{S^{k},\bm{\mu}}(\bm{x}_i) \\
    =&\,\,\,\sum_{\bm{\mu}}
    \underbrace{\biggl(L_{S^{k},\bm{\mu}} + \sum_{\bm{\mu}'}\mathcal{L}_{S^{k-1},\bm{\mu}'}\cdot\ell_{S^{k-1},\bm{\mu}'}(\bm{\xi}_{S^{k},\bm{\mu}})\biggr)}_{=\mathcal{L}_{S^{k},\bm{\mu}}}
    \ell_{S^{k},\bm{\mu}}(\bm{x}_i).
    \qedhere
\end{align*}
\end{proof}

\section{Data Structure of the Cluster Tree}\label{appendx:treecode_clustertree_datastructure}
Although the FMM solvers developed in this work are written in the Julia programming language, we use C-style pseudocode to illustrate the data structure of the cluster tree. 
The data structure of the cluster tree can be naively designed as follow:
\begin{verbatim}
struct Cluster {
  size_t npar;
  value_type (*positions)[3]; // array of particle positions 
  Cluster* children;
}    
\end{verbatim}
However, this naive implementation may require a significant amount of memory as the position of particles in each cluster is explicitly stored. For a balanced cluster tree describing an $N$-particles cluster, the number of particle positions to be stored is $N\log_2{N}$. If we have $N=2\times 10^6$, a memory of roughly $1$~GB will need to be allocated during the construction of the cluster tree and this could cause a performance bottleneck. 

Alternatively, one may store the particle positions outside the structure and declare an external array \texttt{parindices} to store the indices of all the particles. In such a case, the data structure can be expressed as 
\begin{verbatim}
size_t parindices[N]
value_type positions[N][3]
struct Cluster {
  size_t pindex_lo;
  size_t pindex_hi;
  Cluster* children;
}
\end{verbatim}
If the elements of \texttt{parindices} are arranged in such a manner that the indices of the particles in the cluster $S$ occupy in \texttt{parindices} contiguously from $l$-th (\texttt{pindex\_lo}) to $h$-th (\texttt{pindex\_hi}) location, their values in \texttt{parindices} (\emph{i.e.}, their indices) can be expressed as
\begin{equation*}
    p_{l},p_{l+1},\ldots,p_{h}.
\end{equation*}
In the subdivision of $S$, we first determine the splitting coordinate direction $g\in\{x,y,z\}$ from $\text{bbox}(S)$ and permute the elements in \texttt{parindices} that
\begin{equation*}
    p'_{l},\ldots,p'_{\left\lfloor\tfrac{l+h}{2}\right\rfloor},p'_{\left\lfloor\tfrac{l+h}{2}\right\rfloor+1},\ldots
    p'_{h}
\end{equation*}
and

\begin{equation*}
    g_{p'_i} 
    \begin{cases}
    \leq g_{p'_{\left\lfloor\tfrac{l+h}{2}\right\rfloor}}\quad\text{ if } \left\lfloor\tfrac{l+h}{2}\right\rfloor \ge i \ge l,\\
    > g_{p'_{\left\lfloor\tfrac{l+h}{2}\right\rfloor}}\quad\text{ if }\left\lfloor\tfrac{l+h}{2}\right\rfloor < i \leq h.
    \end{cases}
\end{equation*}
This permutation enables the objects of the children clusters \texttt{S1} and \texttt{S2} to access their belonging particle indices by:
\begin{align*}
    & \texttt{S1.pindex\_lo=}l,\, \texttt{S1.pindex\_hi=}\lfloor\tfrac{l+h}{2}\rfloor,\\
    & \texttt{S2.pindex\_lo=}\lfloor\tfrac{l+h}{2}\rfloor+1,\, \texttt{S2.pindex\_hi=}h,\\
    &\texttt{parindicies\string[S1.pindex\_lo\string]},\ldots,\texttt{parindicies\string[S1.pindex\_hi\string]},\\
    &\texttt{parindicies\string[S2.pindex\_lo\string]},\ldots,\texttt{parindicies\string[S2.pindex\_hi\string]}.
\end{align*}
In this study, the permutation is implemented by the Quickselect algorithm with the Lomuto partition scheme~\cite{cormen2022introduction}. The complexity on average is $\mathcal{O}(N)$ and can be $\mathcal{O}(N^2)$ in the worst-case scenario.




\section*{Acknowledgement}
This work was supported by DASHH (Data Science in Hamburg -- HELM\-HOLTZ Graduate School for the Structure of Matter) with the Grant-No.\ HIDSS-0002 and in part by the European Research Council under the European Union's Seventh Framework Programme (FP7/2007-2013) through Synergy Grant (609920). The authors acknowledge the computational resources of the Maxwell Cluster operated at Deutsches Elektronen-Synchrotron (DESY).


\bibliographystyle{elsarticle-num}
\bibliography{references}

\begin{thebibliography}{10}
\expandafter\ifx\csname url\endcsname\relax
  \def\url#1{\texttt{#1}}\fi
\expandafter\ifx\csname urlprefix\endcsname\relax\def\urlprefix{URL }\fi
\expandafter\ifx\csname href\endcsname\relax
  \def\href#1#2{#2} \def\path#1{#1}\fi

\bibitem{zhang2015differential}
H.~Zhang, J.~Portman, Z.~Tao, P.~Duxbury, C.-Y. Ruan, K.~Makino, M.~Berz, {The
  Differential Algebra Based Multiple Level Fast Multipole Algorithm for 3D
  Space Charge Field Calculation and Photoemission Simulation}, Microscopy and
  Microanalysis 21~(S4) (2015) 224--229.
\newblock \href {https://doi.org/10.1017/S1431927615013410}
  {\path{doi:10.1017/S1431927615013410}}.

\bibitem{zhang2017fast}
H.~Zhang, H.~Huang, R.~Li, J.~Chen, L.-S. Luo, Fast multipole method using
  {C}artesian tensor in beam dynamic simulation, AIP Conference Proceedings
  1812~(1) (2017) 050001.
\newblock \href {https://doi.org/10.1063/1.4975862}
  {\path{doi:10.1063/1.4975862}}.

\bibitem{langston2021mach}
M.~Langston, R.~Lethin, P.~Letourneau, M.~Morse, J.~Wei, {MACH-B: Fast
  Multipole Method Approaches in Particle Accelerator Simulations for the
  Computational and Intensity Frontiers}, in: {Proceedings of the 12th
  International Particle Accelerator Conference}, no.~12 in International
  Particle Accelerator Conference, JACoW Publishing, Geneva, Switzerland, 2021,
  pp. 4237--4240.
\newblock \href {https://doi.org/10.18429/JACoW-IPAC2021-THPAB227}
  {\path{doi:10.18429/JACoW-IPAC2021-THPAB227}}.

\bibitem{gordon2021point}
M.~Gordon, S.~Van Der~Geer, J.~Maxson, Y.-K. Kim, {Point-to-point Coulomb
  effects in high brightness photoelectron beam lines for ultrafast electron
  diffraction}, Physical Review Accelerators and Beams 24~(8) (2021) 084202.
\newblock \href {https://doi.org/10.1103/PhysRevAccelBeams.24.084202}
  {\path{doi:10.1103/PhysRevAccelBeams.24.084202}}.

\bibitem{schmid2019simulating}
S.~Schmid, H.~De~Gersem, M.~Dohlus, E.~Gjonaj, {Simulating Space Charge
  Dominated Beam Dynamics Using FMM}, in: Proceedings of 3rd North American
  Particle Accelerator Conference, 2019, p. WEPLE10.
\newblock \href {https://doi.org/10.18429/JACoW-NAPAC2019-WEPLE10}
  {\path{doi:10.18429/JACoW-NAPAC2019-WEPLE10}}.

\bibitem{dawson1983particle}
J.~M. Dawson, {Particle simulation of plasmas}, Reviews of Modern Physics 55
  (1983) 403--447.
\newblock \href {https://doi.org/10.1103/RevModPhys.55.403}
  {\path{doi:10.1103/RevModPhys.55.403}}.

\bibitem{birdsall2018plasma}
C.~K. Birdsall, A.~B. Langdon, {Plasma Physics via Computer Simulation}, CRC
  press, 2018.

\bibitem{flottmann2003recent}
K.~Fl\"{o}ttmann, S.~Lidia, P.~Piot, Recent improvements to the {ASTRA}
  particle tracking code, Tech. rep., Lawrence Berkeley National Lab (LBNL),
  USA (2003).

\bibitem{qiang2006three-dimensional}
J.~Qiang, S.~Lidia, R.~D. Ryne, C.~Limborg-Deprey, Three-dimensional
  quasistatic model for high brightness beam dynamics simulation, Phys. Rev. ST
  Accel. Beams 9 (2006) 044204.
\newblock \href {https://doi.org/10.1103/PhysRevSTAB.9.044204}
  {\path{doi:10.1103/PhysRevSTAB.9.044204}}.

\bibitem{qiang2018symplectic}
J.~Qiang, Symplectic particle-in-cell model for space-charge beam dynamics
  simulation, Phys. Rev. Accel. Beams 21 (2018) 054201.
\newblock \href {https://doi.org/10.1103/PhysRevAccelBeams.21.054201}
  {\path{doi:10.1103/PhysRevAccelBeams.21.054201}}.

\bibitem{jones1998hybrid}
F.~W. Jones, {A hybrid fast-multipole technique for space-charge tracking with
  halos}, AIP Conference Proceedings 448~(1) (1998) 359--370.
\newblock \href {https://doi.org/10.1063/1.56759} {\path{doi:10.1063/1.56759}}.

\bibitem{schmid2018reptil}
S.~A. Schmid, H.~D. Gersem, E.~Gjonaj, {REPTIL} {-} {A} {R}elativistic 3{D}
  {S}pace {C}harge {P}article {T}racking {C}ode {B}ased on the {F}ast
  {M}ultipole {M}ethod, unpublished (01 2019).

\bibitem{qiang2017symplectic}
J.~Qiang, {Symplectic multiparticle tracking model for self-consistent
  space-charge simulation}, Physical Review Accelerators and Beams 20 (2017)
  014203.
\newblock \href {https://doi.org/10.1103/PhysRevAccelBeams.20.014203}
  {\path{doi:10.1103/PhysRevAccelBeams.20.014203}}.

\bibitem{schmid2021energy}
S.~Schmid, H.~D. Gersem, E.~Gjonaj, {Energy-Binning Fast Multipole Method for
  Electron Injector Simulations}, in: {Proceedings of the 12th International
  Particle Accelerator Conference}, no.~12 in International Particle
  Accelerator Conference, JACoW Publishing, Geneva, Switzerland, 2021, pp.
  4244--4246.
\newblock \href {https://doi.org/10.18429/JACoW-IPAC2021-THPAB229}
  {\path{doi:10.18429/JACoW-IPAC2021-THPAB229}}.

\bibitem{fubiani2006space}
G.~Fubiani, J.~Qiang, E.~Esarey, W.~P. Leemans, G.~Dugan, {Space charge
  modeling of dense electron beams with large energy spreads}, Phys. Rev. ST
  Accel. Beams 9 (2006) 064402.
\newblock \href {https://doi.org/10.1103/PhysRevSTAB.9.064402}
  {\path{doi:10.1103/PhysRevSTAB.9.064402}}.

\bibitem{kan2023relativistic}
Y.-K. Kan, F.~X. Kärtner, S.~{Le Borne}, J.-P.~M. Zemke, {Relativistic
  space-charge field calculation by interpolation-based treecode}, Computer
  Physics Communications 286 (2023) 108668.
\newblock \href {https://doi.org/10.1016/j.cpc.2023.108668}
  {\path{doi:10.1016/j.cpc.2023.108668}}.

\bibitem{wilson2021gpu-accelerated}
L.~Wilson, N.~Vaughn, R.~Krasny, {A GPU-accelerated fast multipole method based
  on barycentric Lagrange interpolation and dual tree traversal}, Computer
  Physics Communications 265 (2021) 108017.
\newblock \href {https://doi.org/10.1016/j.cpc.2021.108017}
  {\path{doi:10.1016/j.cpc.2021.108017}}.

\bibitem{Wienke2012openacc}
S.~Wienke, P.~Springer, C.~Terboven, D.~an~Mey, Openacc --- first experiences
  with real-world applications, in: C.~Kaklamanis, T.~Papatheodorou, P.~G.
  Spirakis (Eds.), Euro-Par 2012 Parallel Processing, Springer Berlin
  Heidelberg, Berlin, Heidelberg, 2012, pp. 859--870.

\bibitem{fong2009black}
W.~Fong, E.~Darve, {The black-box fast multipole method}, Journal of
  Computational Physics 228~(23) (2009) 8712--8725.
\newblock \href {https://doi.org/10.1016/j.jcp.2009.08.031}
  {\path{doi:10.1016/j.jcp.2009.08.031}}.

\bibitem{wang2020kernel}
L.~Wang, R.~Krasny, S.~Tlupova, A kernel-independent treecode based on
  barycentric {L}agrange interpolation, Communications in Computational Physics
  28~(4) (2020) 1415--1436.
\newblock \href {https://doi.org/10.4208/cicp.OA-2019-0177}
  {\path{doi:10.4208/cicp.OA-2019-0177}}.

\bibitem{wilson2021development}
L.~Wilson, {Development and Application of Numerical Methods in Biomolecular
  Solvation}, Ph.D. thesis, University of Michigan (2021).
\newblock \href {https://doi.org/10.7302/1547} {\path{doi:10.7302/1547}}.

\bibitem{boerm2010efficient}
S.~B\"{o}rm, Efficient numerical methods for non-local operators:
  ${H}^2$-matrix compression, algorithms and analysis, Vol.~14 of EMS Tracts in
  Mathematics, European Mathematical Society (EMS), Z\"{u}rich, 2010.
\newblock \href {https://doi.org/10.4171/091} {\path{doi:10.4171/091}}.

\bibitem{hackbusch2015hierarchical}
W.~Hackbusch, {Hierarchical Matrices: Algorithms and Analysis}, Vol.~49,
  Springer, 2015.
\newblock \href {https://doi.org/10.1007/978-3-662-47324-5}
  {\path{doi:10.1007/978-3-662-47324-5}}.

\bibitem{appel1985efficient}
A.~W. Appel, An efficient program for many-body simulation, SIAM Journal on
  Scientific and Statistical Computing 6~(1) (1985) 85--103.
\newblock \href {https://doi.org/10.1137/0906008} {\path{doi:10.1137/0906008}}.

\bibitem{dehnen2002hierarchical}
W.~Dehnen, {A Hierarchical O(N) Force Calculation Algorithm}, Journal of
  Computational Physics 179~(1) (2002) 27--42.
\newblock \href {https://doi.org/10.1006/jcph.2002.7026}
  {\path{doi:10.1006/jcph.2002.7026}}.

\bibitem{bezanson2017julia}
J.~Bezanson, A.~Edelman, S.~Karpinski, V.~B. Shah, {Julia: A Fresh Approach to
  Numerical Computing}, SIAM Review 59~(1) (2017) 65--98.
\newblock \href {https://doi.org/10.1137/141000671}
  {\path{doi:10.1137/141000671}}.

\bibitem{liu2019efficient}
J.~Liu, M.~Robson, T.~Quinn, M.~Kulkarni, {Efficient GPU Tree Walks for
  Effective Distributed N-Body Simulations}, in: Proceedings of the ACM
  International Conference on Supercomputing, ICS '19, Association for
  Computing Machinery, New York, NY, USA, 2019, p. 24–34.
\newblock \href {https://doi.org/10.1145/3330345.3330348}
  {\path{doi:10.1145/3330345.3330348}}.

\bibitem{hwu2011GPU_ch6}
M.~Burtscher, K.~Pingali, {Chapter 6 - An Efficient CUDA Implementation of the
  Tree-Based Barnes Hut n-Body Algorithm}, in: W.~mei W.~Hwu (Ed.), {GPU
  Computing Gems Emerald Edition}, Applications of GPU Computing Series, Morgan
  Kaufmann, Boston, 2011, pp. 75--92.
\newblock \href {https://doi.org/10.1016/B978-0-12-384988-5.00006-1}
  {\path{doi:10.1016/B978-0-12-384988-5.00006-1}}.

\bibitem{wilson2022private_communication}
{Leighton Wilson}, private communication (2022).

\bibitem{BaryTree}
L.~Wilson, N.~Vaughn, \href{https://github.com/Treecodes/BaryTree}{{BaryTree}}
  (2021).
\newline\urlprefix\url{https://github.com/Treecodes/BaryTree}

\bibitem{cheng2014professional}
J.~Cheng, M.~Grossman, T.~McKercher, {Professional CUDA C Programming}, John
  Wiley \& Sons, 2014.

\bibitem{staticarrays}
A.~Ferris, other contributors,
  \href{https://github.com/JuliaArrays/StaticArrays.jl}{{StaticArrays}} (2016).
\newline\urlprefix\url{https://github.com/JuliaArrays/StaticArrays.jl}

\bibitem{cudajl}
{Julia Computing}, other contributors,
  \href{https://github.com/JuliaArrays/StaticArrays.jl}{{CUDA.jl}} (2016).
\newline\urlprefix\url{https://github.com/JuliaArrays/StaticArrays.jl}

\bibitem{besard2019effective}
T.~Besard, C.~Foket, B.~D. Sutter, {Effective Extensible Programming:
  Unleashing Julia on GPUs}, {IEEE} Transactions on Parallel and Distributed
  Systems 30~(4) (2019) 827--841.
\newblock \href {https://doi.org/10.1109/tpds.2018.2872064}
  {\path{doi:10.1109/tpds.2018.2872064}}.

\bibitem{cormen2022introduction}
T.~H. Cormen, C.~E. Leiserson, R.~L. Rivest, C.~Stein, Introduction to
  algorithms, 4th Edition, MIT press, 2022.

\end{thebibliography}







\end{document}